\newif\ifdraft
\draftfalse
\documentclass[letterpaper,11pt]{amsart}
\usepackage{amsmath}
\usepackage{amsthm}
\usepackage{amssymb}
\usepackage{amsfonts}
\usepackage[top=1.0in, bottom=1.0in,left=1.0in, right=1.0in,marginparwidth=0.8in,marginparsep=3mm]{geometry}
\usepackage{fullpage}
\usepackage{graphicx,color}
\usepackage[dvipsnames]{xcolor}
\usepackage{fancybox}
\usepackage{tikz}
\usepackage{cite}
\usepackage[utf8]{inputenc}
\usepackage{bm}
\renewcommand{\mathbf}[1]{\bm{#1}}

\usepackage{cool}
\Style{DSymb={\mathrm d},DShorten=false,IntegrateDifferentialDSymb=\mathrm{d}}

\ifdraft
\usepackage{todonotes}
\usepackage{scrtime}
\usepackage{ifthen}
\usepackage{prelim2e}
\newcommand{\datetime}{\the\year-\ifthenelse{\the\month < 10}{0}{}\the\month-\ifthenelse{\the\day < 10}{0}{}\the\day{} \thistime}

\fi
\ifdraft
\usepackage[pagebackref]{hyperref}
\else
\usepackage{hyperref}
\fi
\hypersetup{
  colorlinks=true,
  linkcolor=green!30!black,
  linktoc=all,
  citecolor=blue,
  bookmarksnumbered=true
}
\newcommand{\arxiv}[1]{\href{https://arxiv.org/abs/#1}{arXiv:#1}}

\newcommand{\ignore}[1]{}

\newtheorem{theorem}{Theorem}[section]

\newtheorem{corollary}[theorem]{Corollary}
\newtheorem{lemma}[theorem]{Lemma}

\theoremstyle{definition}
\newtheorem{definition}[theorem]{Definition}
\newtheorem*{remark}{Remark}
\newtheorem*{remarks}{Remarks}

\newcommand{\abs}[1]{\left\vert#1\right\vert}
\newcommand{\set}[1]{\left\{#1\right\}}
\newcommand{\eps}{\varepsilon}

\newcommand{\defeq}{:=}

\newcommand{\poly}[1]{\ensuremath{\mathop{\mathrm{poly}}\inp{#1}}}
\newcommand{\polylog}[1]{\ensuremath{\mathop{\mathrm{polylog}}\inp{#1}}}

\newcommand{\CommentS}[1]{}

\newcommand{\T}{\mathcal{T}}
\newcommand{\G}{\mathcal{G}}
\newcommand{\Gk}{\G^{\Delta,k}_t}
\newcommand{\Gi}{\G^{\Delta,k}_i}
\newcommand{\Ck}{\mathcal{C}_t^{\Delta,k}}
\newcommand{\Ci}{\mathcal{C}_i^{\Delta,k}}
\newcommand{\C}{\mathbb{C}}
\newcommand{\R}{\mathbb{R}}
\newcommand{\N}{\mathbb{N}}
\newcommand{\Ot}{\widetilde{O}}

\usepackage{mathtools}

\DeclarePairedDelimiter\floor{\lfloor}{\rfloor}

\newcommand{\st}[1][]{\ensuremath{\;\mathbf{#1\vert}\;}}

\renewcommand{\vec}[1]{\ensuremath{\mathbf{#1}}}

\newcommand{\higraph}[0]{insect}
\newcommand{\higraphs}[0]{insects}

\newcommand{\Higraphs}[0]{Insects}
\renewcommand{\a}[0]{an}
\newcommand{\A}[0]{An}
\newcommand{\sub}[0]{sub-}
\newcommand{\ind}[2]{\ensuremath{#1^{+}\left[{#2}\right]}}
\newcommand{\subinsect}[2]{\ensuremath{#1 \hookrightarrow #2}}
\newcommand{\isubinsect}[2]{\ensuremath{\mathbf{1}\inb{\subinsect{#1}{#2}}}}

\usepackage[tracking]{microtype}
\usepackage[ttscale=0.81]{libertine}
\usepackage[T1]{fontenc}

\begin{document}
\title{The Ising Partition Function: \\Zeros and Deterministic Approximation}
\author{Jingcheng Liu \and Alistair Sinclair \and Piyush Srivastava}

\date{}
\maketitle
\begin{abstract}

  We study the problem of approximating the partition function of the
  ferromagnetic Ising model with both pairwise as well as higher order
  interactions (equivalently, in graphs as well as hypergraphs).  Our
  approach is based on the classical Lee-Yang theory of phase
  transitions, along with a new Lee-Yang theorem for the Ising model
  with higher order interactions, and on an extension of ideas
  developed recently by Barvinok, and Patel and Regts that can be seen
  as an algorithmic realization of the Lee-Yang theory.

  Our first result is a \emph{deterministic} polynomial time approximation scheme (an
  FPTAS) for the partition function in bounded degree graphs that is
  valid over the entire range of parameters $\beta$ (the interaction)
  and $\lambda$ (the external field), except for the case
  $\abs{\lambda}=1$ (the ``zero-field'' case).  A polynomial time \emph{randomized}
  approximation scheme (FPRAS) for all graphs and all $\beta,\lambda$, based on
  Markov chain Monte Carlo simulation, has long
  been known.  Unlike most other deterministic approximation
  algorithms for problems in statistical physics and counting, our
  algorithm does not rely on the ``decay of correlations'' property,
  but, as pointed out above, on Lee-Yang theory.  This approach
  extends to the more general setting of the Ising model on
  hypergraphs of bounded degree and edge size, where no previous
  algorithms (even randomized) were known for a wide range of
  parameters.  In order to achieve this latter extension, we establish
  a tight version of the Lee-Yang theorem for the Ising model on
  hypergraphs, improving a classical result of Suzuki and Fisher.

\par\bigskip\bigskip\par\noindent
\end{abstract}

\thispagestyle{empty}
\let\bakthefootnote\thefootnote
\let\thefootnote\relax
\footnotetext{Jingcheng Liu, Computer Science Division, UC Berkeley. Email: \texttt{liuexp@berkeley.edu}.}
\footnotetext{Alistair Sinclair, Computer Science Division, UC Berkeley. Email: \texttt{sinclair@cs.berkeley.edu}.}
\footnotetext{Piyush Srivastava, Tata Institute of Fundamental Research, Mumbai. Email:
  \texttt{piyush.srivastava@tifr.res.in}.}
  \footnotetext{An extended abstract of this paper appeared in the
    proceedings of the \emph{58th Annual IEEE Symposium on Foundations
      of Computer Science (FOCS)}, 2017. pp.986--997.}
\let\thefootnote\bakthefootnote
\newpage
\setcounter{page}{1}

\section{Introduction}
The Ising model, first studied a century ago as a model for magnetic materials
by Lenz and Ising~\cite{isi25}, has since become an
important tool for the modeling of interacting systems.  In the Ising
model, such a system is represented as a graph $G = (V, E)$, so that
the individual entities comprising the system correspond to the vertices $V$ and
their pairwise interactions to the edges~$E$.
A \emph{configuration} of the system is an assignment
$\bm \sigma: V \rightarrow \inbr{+,-}$ of one of two possible values (often called
``spins'') to each vertex.
The model then induces a
probability distribution (known as a \emph{Gibbs distribution}) over
these global configurations in terms of local parameters that model
the interactions of the vertices.

In our setting, it will be convenient to parameterize these
interactions in terms of a \emph{vertex activity} or
\emph{fugacity}~$\lambda$, that models an ``external field'' and
determines the propensity of a vertex to be in the $+$ configuration,
and an \emph{edge activity} $\beta\ge 0$ that models the tendency of
vertices to agree with their
neighbors.  The model assigns to each configuration~$\bm\sigma$ a weight
\[
w(\bm \sigma) \defeq \beta^{\abs{\inbr{\inbr{u, v} \in E \st \bm\sigma(u) \neq \bm\sigma(v)}}} \lambda^{\abs{\inbr{v \st \bm\sigma(v) = +}}}
= \beta^{|{E\inp{S, \overline{S}}}|}  \lambda^{\abs{S}},
\]
where $S=S(\bm\sigma)$ is the set of vertices assigned spin~+
in~$\bm\sigma$ and $E\inp{S, \overline{S}}$ is the set of edges in the cut
$\inp{S,\overline{S}}$ (i.e., the number of pairs of adjacent vertices assigned
opposite spins).
The probability of configuration $\bm \sigma$ under the Gibbs distribution is then
$\mu(\bm\sigma):=w(\bm \sigma)/{Z_G^\beta( \lambda)}$,
where the normalizing factor $Z_G^\beta( \lambda)$ is the \emph{partition function} defined as
\begin{equation}\label{eqn:isingdef}
  Z_G^\beta( \lambda) \defeq \sum_{\bm \sigma : V \rightarrow \inbr{+, -}}
  w(\bm \sigma) = \sum_{S \subseteq V} \beta^{|{E\inp{S, \overline{S}}}|}  \lambda^{\abs{S}}.
\end{equation}
Notice that the partition function may be interpreted combinatorially as
a cut generating polynomial in the graph~$G$.

In this paper we focus on the original \emph{ferromagnetic} case in
which $\beta < 1$, so that configurations in which a larger number of
neighboring spins agree (small cuts) have higher probability.  The
\emph{anti-ferromagnetic} regime $\beta > 1$ is qualitatively very
different, and prefers configurations with disagreements
between neighbors.  We note also that all our results in this paper hold in
the more general setting where there is a different
interaction~$\beta_e$ on each edge, provided that all the $\beta_e$
satisfy whatever constraints we are putting on~$\beta$.  (So, e.g., in
the ferromagnetic case $\beta_e<1$ for all~$e$.)  In addition,
our results allow $\beta$ to be negative and
$\lambda$ to be complex; this will be discussed in more detail below.

As in almost all statistical physics and graphical models, the
partition function captures the computational complexity of the Ising
model.  Partition functions are
\#P-hard\footnote{If a combinatorial counting problem, such as
  computing a partition function in a statistical physics model, is
  \#P-hard, then it can be solved in polynomial time only if
  \emph{all} counting problems belonging to a very rich class can be
  solved in polynomial time.  Hence \#P-hardness is widely regarded as
  compelling evidence of the intractibility of efficient exact
  computation.  For a more detailed account of this phenomenon in the
  context of partition functions, see, e.g.,~\cite[Appendix
  A]{sinclair12:_lee_yang}.}  to compute exactly in virtually any
interesting case (e.g., this is true for the Ising model except in the
trivial cases $\lambda=0$ or $\beta\in\{0,1\}$), so attention is
focused on approximation.  An early result in the field due to Jerrum
and Sinclair~\cite{jersin93} establishes a {\it fully polynomial
  randomized approximation scheme\/} for the Ising partition function,
valid for all graphs~$G$ and all values of the parameters
$(\beta,\lambda)$ in the ferromagnetic regime.  Like many of the first
results on approximating partition functions, this algorithm is based
on random sampling and the Markov chain Monte Carlo method.

For several statistical physics models on bounded degree graphs
(including the anti-ferromagnetic Ising
model~\cite{sinclair_approximation_2012,li_correlation_2011} and the
``hard core'', or independent set model~\cite{Weitz}), fully-polynomial
{\it deterministic\/} approximation schemes have since been developed, based
on the decay of correlations property in those models: roughly
speaking, one can estimate the local contribution to the partition
function at a given vertex~$v$ by exhaustive enumeration in a
neighborhood around~$v$, using decay of correlations to truncate the
neighborhood at logarithmic diameter.  The range of applicability of
these algorithms is of course limited to the regime in which decay of
correlations holds, and indeed complementary results prove that the
partition function is NP-hard to approximate outside this
regime~\cite{sly12, Vigoda-hard-core-11}.  Perhaps surprisingly,
however, no deterministic approximation algorithm is known for the
classical ferromagnetic Ising partition function that works over
anything close to the full range of the randomized
algorithm of~\cite{jersin93}: to the best of our knowledge, the best
deterministic algorithm, due to Zhang, Liang and Bai~\cite{zhaliabai09},
is based on correlation decay and is applicable to graphs of maximum
degree $\Delta$ only when $\beta > 1 - 2/\Delta$.

The restricted applicability of correlation decay based algorithms
 for the ferromagnetic Ising model arises from two related
 reasons: the first is that this model does not exhibit correlation decay
 for $\beta$ sufficiently close to $0$ (for any
given value of the external field), so any straightforward
 approach based only on this property cannot be expected to
 work for all~$\beta$. Secondly,
there is a regime of parameters for which, even though decay of
correlation holds, there is evidence to believe that it cannot be
exploited to give an algorithm using the usual techniques: see
\cite[Appendix 2]{sinclair_approximation_2012} for a more detailed
discussion of this point.

The first goal of this paper is to supply such a deterministic algorithm which covers almost
the entire range of parameters of the model except for the ``zero-field'' case $|\lambda|=1$:
\begin{theorem}\label{thm:main1}
  Fix any $\Delta>0$.  There is a fully polynomial
  time approximation scheme (FPTAS)\footnote{A FPTAS takes as input an $n$-vertex
  (hyper)graph~$G$, model parameters $\beta,\lambda$, 
  and an accuracy parameter~$\varepsilon\in(0,1)$ and outputs
  a value that approximates $Z_G^\beta(\lambda)$ within a factor
  $1\pm\varepsilon$ (see also eq.~\eqref{eq:fptasdef}).
  The running time of the algorithm is polynomial in~$n$ and~$1/\varepsilon$.}
  for the Ising partition function
  $Z_G^\beta(\lambda)$ in all graphs~$G$ of maximum degree~$\Delta$ for all
  edge activities $-1\le\beta\le 1$ and all (possibly complex) vertex
  activities~$\lambda$ with $|\lambda|\ne 1$.
\end{theorem}
\begin{remarks}
  (i) For fixed $\Delta$ and $\lambda$ such that $\abs{\lambda} < 1$,
  the running time of the FPTAS for producing a
  $(1\pm\varepsilon)$-factor approximation on $n$-vertex
  graphs of degree at most $\Delta$ is
  $(n/\varepsilon)^ {O\inp{\frac{\log \Delta}{\abs{1 -
          \abs{\lambda}}}}}$
  (the running times of the algorithms in
  Theorems\nobreakspace \ref {thm:main2} and\nobreakspace  \ref {thm:main3} have a similar dependence on $\lambda$
  and $\Delta$).  Such dependence on the ``distance to the critical
  boundary'' (in this case, the circle $\abs{\lambda} = 1$) of the
  degree of the polynomial bounding the running time of the FPTAS
  appears to be a common feature of algorithms based on correlation
  decay~\cite{Weitz,SSSY15,li_correlation_2011} as well as our present
  analytic continuation approach.  In contrast, approximate counting
  algorithms based on Markov chain Monte Carlo (e.g.,
  \cite{jerrum_approximating_1989,efthymiou2016convergence,luby_approximately_1997})
  often have the desirable feature that they are in a sense ``fixed
  parameter tractable'': even as the fixed parameters of the problems
  are moved close to the boundary up to which the algorithm is
  applicable, the degree of the polynomial bounding their running time
  does not increase; it is only the constant factors which increase to
  infinity. (ii) Note that although $\lambda, \beta$ are positive in
  the ``physically relevant'' range in most applications of the Ising
  model, the above theorem also applies more generally to
  $\beta \in [-1, 1]$ and complex valued $\lambda$ not on the unit
  circle.  Moreover, we can allow edge-dependent activities~$\beta_e$
  provided all of them lie in $[-1,1]$.
\end{remarks}
The above theorem is actually a special case of a more general theorem
for the hypergraph version of the Ising model (\MakeUppercase Theorem\nobreakspace \ref {thm:main2} below).
We now illustrate our approach to proving these theorems, which will also allow us to introduce and motivate our
further results in the paper.

In contrast to previous algorithms based on correlation decay, our
algorithm is based on isolating the complex zeros of the partition
function~$Z:=Z_G^\beta(\lambda)$ (viewed as a polynomial in~$\lambda$
for a fixed value of~$\beta$).
This approach was introduced recently by
Barvinok~\cite{Barvinok15,barvinok_computing_2015} (see also the
recent monograph by
Barvinok~\cite{barvinok16:_combin_compl_partit_funct} for a discussion
of the approach in a more
general context).  We start with the observation
that the $1\pm\varepsilon$ multiplicative approximation of $Z$ required
for a FPTAS corresponds to a $O(\varepsilon)$ additive approximation of $\log Z$.
Barvinok's approach considers a Taylor expansion of $\log Z$ around a
point $\lambda_0$ such that $Z(\lambda_0)$ is easy to evaluate.  (For
the Ising model, $\lambda_0 = 0$ is such a choice.)  It then seeks
to evaluate the function at other points by carrying out an explicit
analytic continuation.  More concretely, suppose it can
be shown that there are no zeros of $Z$ in the open disk
$D(\lambda_0, r)$ of radius~$r$ around $\lambda_0$.  From standard
results in complex analysis, it then follows that the Taylor expansion
around $\lambda_0$ of $\log Z$ is absolutely convergent in
$D(\lambda_0, r)$ and further, that the first $m$ terms of this
expansion evaluated at a point $\lambda \in D(\lambda_0, r)$ provide a
$O\inp{\frac{n \alpha^{m}}{1-\alpha}}$ additive approximation of $\log Z(\lambda)$,
where $\alpha = \abs{\lambda-\lambda_0}/r < 1$, and $n$  is the degree
of $Z$ as a polynomial in $\lambda$.
We note that Barvinok's
approach may be seen as an algorithmic counterpart of the
traditional view of phase transitions in statistical physics
in terms of analyticity of $\log Z$~\cite{leeyan52}.

To apply this approach in the case of the ferromagnetic Ising model,
we may appeal to the famous Lee-Yang theorem of the
1950s~\cite{leeyan52b}, which establishes that the zeros of
$Z(\lambda)$ all lie on the unit circle in the complex plane.  This
allows us to take $\lambda_0=0$ and $r=1$ in the previous paragraph,
and thus approximate $Z(\lambda)$ at any point $\lambda$
satisfying $\abs{\lambda} < 1$.
This extends to all $\lambda$ with $|\lambda|\ne 1$ via the fact that
$Z(\lambda) = \lambda^n Z(\frac{1}{\lambda})$.
\begin{remark}
We note that the case $|\lambda|=1$ is likely to require additional ideas
because it is known that there exist bounded degree graphs (namely $\Delta$-ary trees)
for which the partition function $Z_G^\beta(\lambda)$ has complex zeros within
distance $O(1/n)$ of $\lambda = 1$, where $n$ is the size of the
graph.  In fact, the zeros are even known to become dense on the
whole unit circle as $n$ increases to
infinity~\cite{barata_distribution_2001,barata_griffiths_1997}.  This
precludes the possibility of efficiently carrying out the analytic continuation
procedure for $\log Z$ to arbitrary points on the unit circle, and
to the point $\lambda=1$ in particular.
\end{remark}

Converting the above approach into an algorithm requires computing the first $k$
coefficients in the Taylor expansion of $\log Z$ around $\lambda_0$.
Barvinok showed that this computation can in turn be reduced to
computing the $O(k)$ lowest-degree  coefficients of the
partition function itself.  In the case of the Ising model, computing
$k$ such coefficients is roughly analogous to computing
$k$-wise correlations between the vertex spins, and doing so
naively on a graph of $n$ vertices requires time $\Omega(n^k)$.  Until
recently, no better ways to perform this calculation were known and
hence approximation algorithms using this approach typically required
quasi-polynomial time\footnote{A {\it quasi-polynomial time\/} algorithm is
one which runs in time $\exp\{O((\log n)^c)\}$ for some constant $c>1$.}
in order to achieve a
$(1 \pm 1/\poly{n})$-factor multiplicative approximation of $Z$
(equivalently, a $1/\poly{n}$ additive approximation of
$\log Z$), since this requires the Taylor series for $\log Z$ to
be evaluated to $k=\Omega(\log n)$
terms~\cite{BarvinokSoberon16a,BarvinokSoberon16b,barvinok_computing_2015}.

Recently, Patel and Regts~\cite{patel_deterministic_2016} proposed a way to get
around this barrier for several classes of partition functions.  Their
method is based on writing the coefficients in the Taylor series of
$\log Z$ as linear combinations of counts of connected induced
subgraphs of size up to $\Theta(\log n)$.  This is already promising,
since the number of connected induced subgraphs of size
$O(\log n)$ of a graph $G$ of maximum degree $\Delta$ is polynomial in the
size of $G$ when $\Delta$ is a fixed constant.  Further, the count of
induced copies of such a subgraph can also be computed in time
polynomial in the size of~$G$~\cite{borgs_left_2013}.  Patel and
Regts built on these tools to show a way to compute the
$\Theta(\log n)$ Taylor coefficients of $\log Z$ needed in Barvinok's
approach for several families of partition functions, for some of
which correlation decay based algorithms are still not known.

Unfortunately, for the case of the Ising model, it is not clear how to
write the Taylor coefficients in terms of induced subgraph
counts. Indeed, in their paper~\cite[Theorem 1.4]{patel_deterministic_2016},
Patel and Regts apply their method to the Ising model viewed as
a polynomial in~$\beta$ rather than~$\lambda$, which allows them to use
subgraph counts.  However, this prevents them from using the Lee-Yang
theorem, and they are consequently able to establish only a small
zero-free region.  As a result, they can handle only the zero-field
``high-temperature'' regime (where in fact the correlation decay property
also holds), specifically the regime $\abs{\beta - 1} \le 0.34 / \Delta$ and $\lambda=1$.

In this paper, we instead propose a generalization of the framework of
Patel and Regts to labelled hypergraphs via objects that we call
\emph{\higraphs{}}.  In the special case of graphs, \a{}
\higraph{} can be seen as a graph that includes edges to additional
boundary vertices:
we refer to Section~\ref{sec:higraphs-hypergraph} for precise definitions.
Using the appropriate notions for counting induced \sub\higraphs{}, we
are then able to write the coefficients arising in the Taylor
expansion of $\log Z$ for the Ising model in terms of induced
\sub\higraph{} counts, and derive from there algorithms for computing
$\Omega(\log n)$ such coefficients in polynomial time in graphs of
bounded degree.  This establishes Theorem~\ref{thm:main1}.  We note
that if one is only interested in deriving \MakeUppercase Theorem\nobreakspace \ref {thm:main1}, then this
can also be done using the notion of \emph{fragments}, developed by Patel
and Regts~\cite{patel_deterministic_2016} in the different context of
edge coloring models, which turns out to be a special case of our
notion of \higraphs{}.

Our framework of \higraphs{}, however, also allows us to extend the
above approach to edge-dependent activities and, more significantly,
to the much more general setting where $G$ is a hypergraph,
as studied, for example, in the classical work of Suzuki and
Fisher~\cite{suzuki1971zeros}, and also more recently in the
literature on approximate
counting~\cite{galanis_complexity_2016,lu16:_fptas_hardc,song16:_count}. 
In a hypergraph of edge size~$k\ge 3$, the pairwise interactions in 
the standard Ising model are replaced by higher-order interactions
of order~$k$.
We note that the Jerrum-Sinclair MCMC approach~\cite{jersin93} apparently
does not extend to hypergraphs, and there is currently no known
polynomial time approximation algorithm (even randomized) for a wide
range of~$\beta$ in this setting.  For a hypergraph $H=(V,E)$,
configurations are again assignments of spins to the vertices~$V$ and
the partition function $Z_H^{\beta}(\lambda)$ is defined exactly as
in~(\ref{eqn:isingdef}), where the cut $E(S,\overline S)$ consists of
those hyperedges with at least one vertex in each of~$S$
and~$\overline S$.  The computation of coefficients via \higraphs{}
carries through as before, but the missing ingredient is an extension
of the Lee-Yang theorem to hypergraphs.  Suzuki and
Fisher~\cite{suzuki1971zeros} prove a weak version of the Lee-Yang
theorem for hypergraphs (see Theorem~\ref{thm:suzuki-fisher} in
section~\ref{sec:leeyang}), which we are able to strengthen to obtain
the following optimal statement, which is of independent interest:

\begin{theorem}\label{thm:leeyangintro}
  Let $H=(V,E)$ be a hypergraph with maximum hyperedge size~$k\ge 3$.
  Then all the zeros of the Ising model partition function
  $Z_H^\beta(\lambda)$ lie on the unit circle if the edge activity
  $\beta$ lies in the range
  $-\frac{1}{2^{k-1} -1} \le \beta \le
  \frac{1}{2^{k-1}\cos^{k-1}\left(\frac{\pi}{k-1} \right) +1}$.
  Further, when $\beta \neq 1$ does not lie in this range, there
  exists a hypergraph $H$ with maximum hypergraph edge size at most
  $k$ such that the zeros of the Ising partition function
  $Z_H^\beta(\lambda)$ of $H$ do not lie on the unit circle.
\end{theorem}
\begin{remark}
  Once again, we can allow edge-dependent activities~$\beta_e$
  provided all of them lie in the range stipulated above.  This
  extension also applies to \MakeUppercase Theorem\nobreakspace \ref {thm:main2} below.
\end{remark}
Note that the classical Lee-Yang theorem (for the graph case $k=2$)
establishes this property for $0\le \beta \le 1$ (the ferromagnetic
regime).  Our proof of Theorem~\ref{thm:leeyangintro} follows along
the lines of Asano's inductive proof of the Lee-Yang
theorem~\cite{asano_lee-yang_1970}, but we need to carefully analyze
the base case (where $H$ consists of a single hyperedge) in order to
obtain the above bounds on~$\beta$.  The optimality of the range of~$\beta$ in
our result follows essentially from the fact that our analysis of the base case is tight.  
For a detailed comparison with the Suzuki-Fisher theorem, see the Remark
following \MakeUppercase Corollary\nobreakspace \ref {cor:hyper-leeyang}.

Combining Theorem~\ref{thm:leeyangintro} with our earlier algorithmic approach immediately
yields the following generalization of Theorem~\ref{thm:main1} to hypergraphs.
\begin{theorem}\label{thm:main2}
Fix any $\Delta>0$ and $k\ge 3$.
There is an FPTAS for the
Ising partition function $Z_H^\beta(\lambda)$ in all hypergraphs~$H$ of maximum
degree~$\Delta$ and maximum edge size~$k$, for all edge activities~$\beta$ in
the range of Theorem~\ref{thm:leeyangintro} and all vertex activities $|\lambda|\ne 1$.
\end{theorem}

Finally, we extend our results to general ferromagnetic two-spin systems on hypergraphs,
again as studied in~\cite{suzuki1971zeros}.
A \emph{two-spin system} on a hypergraph $H=(V,E)$ is specified by hyperedge
activities $\varphi_e: \set{+,-}^{\abs{e}} \to \C$ for $e\in E$, and a vertex activity
$\psi: \set{+,-} \to \C$.
(Note that we treat each hyperedge~$e$ as a set of vertices.)
Then the partition function is defined as:
\begin{equation*}
  Z_H^{\bm \varphi, \bm\psi} \defeq \sum_{\bm \sigma: V \to \set{+,-}} \prod_{ e \in E} \varphi_e\inp{\bm \sigma\big|_e} \prod_{v\in V} \psi(\bm\sigma(v)).
\end{equation*}
Without loss of generality, we will henceforth assume that
$\varphi_e(-,\cdots, -) = 1$, and that $\psi(-) = 1$, $\psi(+)=\lambda$.
We can then write the partition function as
\begin{equation}
  \label{eqn:hypergraph2spin}
  Z_H^{\bm \varphi}(\lambda) =\sum_{\bm \sigma: V \to \set{+,-}} \prod_{ e \in E} \varphi_e\inp{\bm \sigma\big|_e} \lambda^{\abs{\set{v:\bm\sigma(v) = +}}}.
\end{equation}

We call a hypergraph two-spin system \emph{symmetric} if $\varphi_e(\bm \sigma) = \overline{\varphi_e(-\bm \sigma)}$.  Suzuki and Fisher~\cite{suzuki1971zeros} proved a
Lee-Yang theorem for symmetric hypergraph two-spin systems (which is
weaker than our Theorem~\ref{thm:leeyangintro} above when specialized to the Ising model).
Combining this with our general algorithmic approach yields our final result:
\begin{theorem}\label{thm:main3}
Fix any $\Delta>0$ and $k\ge 2$ and a family of symmetric edge activities $\bm\varphi=\{\varphi_e\}$ satisfying
$\abs{\varphi_e(+,\cdots, +)} \ge \frac{1}{4}\sum_{\bm \sigma\in \set{+,-}^V} \abs{\varphi_e(\bm \sigma)}$.
Then there exists an FPTAS for the partition function $Z_H^{\bm \varphi}(\lambda)$ of the corresponding
symmetric hypergraph two-spin system in all hypergraphs~$H$
of maximum degree~$\Delta$ and maximum edge size~$k$ for all vertex activities
$\lambda\in \C$ such that $\abs{\lambda}\neq 1$.
\end{theorem}

The remainder of the paper is organized as follows.  In
section~\ref{sec:taylor}, we spell out Barvinok's approach to
approximating partition functions using Taylor series.
Section~\ref{sec:coefficients} introduces the notion of \higraphs{} and
shows how to use them to efficiently compute the lowest-degree coefficients
of the partition function in the general context of hypergraphs; as
discussed above, this machinery applied to graphs, in conjunction with
the Lee-Yang theorem, implies \MakeUppercase Theorem\nobreakspace \ref {thm:main1}.  Finally, in
section~\ref{sec:leeyang} we prove our extension of the Lee-Yang
theorem to the hypergraph Ising model
(Theorem~\ref{thm:leeyangintro}), and then use it and the
Suzuki-Fisher theorem to prove our algorithmic results for
hypergraphs, Theorems~\ref{thm:main2} and~\ref{thm:main3}.

\subsection{Related work}

The problem of computing partition functions has been widely studied, not only in
statistical physics but also in combinatorics, because the partition function is often
a generating function for combinatorial objects (cuts, in the case of the Ising model).
There has been much progress on \emph{dichotomy theorems}, which attempt to
completely classify these problems as being either \#P-hard or computable (exactly)
in polynomial time (see, e.g., \cite{cai_graph_2010,golgrojerthu09}).

Since the problems are in fact \#P-hard in most cases, algorithmic interest has focused
largely on {\it approximation}, motivated also by the general observation that
approximating the partition function is polynomial time equivalent to sampling
(approximately) from the underlying Gibbs distribution\cite{jervalvaz86}.
In fact, most early approximation algorithms exploited this connection, and gave
{\it fully-polynomial randomized approximation schemes (FPRAS)\/} for the
partition function using Markov chain Monte Carlo (MCMC) samplers for the
Gibbs distribution.  In particular, for the ferromagnetic Ising model, the
MCMC-based algorithm of Jerrum and Sinclair~\cite{jersin93} is valid for all
positive real values of $\lambda$ and for all graphs, irrespective of their vertex degrees.
(For the connection with random sampling in this case, see~\cite{randall_sampling_1999}.)
This was later extended to ferromagnetic two-spin systems by Goldberg, Jerrum
and Paterson~\cite{goldberg2003computational}.  Similar techniques have been
applied recently to the related random-cluster model by Guo and Jerrum~\cite{guo2017random}.

Much detailed work has been done on MCMC for Ising
spin configurations for several important classes of graphs, including two-dimensional
lattices (e.g.,~\cite{martinelli1994approach1,martinelli1994approach2,Lubetzky2012}),
random graphs and graphs of bounded degree (e.g.,~\cite{mossel2013exact}),
the complete graph (e.g.,~\cite{long2011power}) and trees
(e.g.,~\cite{Berger2005,Martinelli2004}); we do not attempt to give a
comprehensive summary of this line of work here.

In the {\it anti-ferromagnetic\/} regime ($\beta>1$), {\it deterministic\/}
approximation algorithms based on correlation decay have been remarkably successful
for graphs of bounded degree.
Specifically, for any fixed integer $\Delta\ge 3$, techniques of Weitz~\cite{Weitz} give
a (deterministic) FPTAS for the anti-ferromagnetic Ising partition function
on graphs of maximum degree~$\Delta$
throughout a region $R_\Delta$ in the $(\beta, \lambda)$ plane (corresponding to the regime
of uniqueness of the Gibbs measure on the $\Delta$-regular
tree)~\cite{sinclair_approximation_2012,li_correlation_2011}.
A complementary result of Sly and Sun~\cite{sly12} (see also~\cite{Vigoda-hard-core-11})
shows that the problem is NP-hard outside~$R_\Delta$.
In contrast, no MCMC based algorithms are known to provide an FPRAS for the
anti-ferromagnetic Ising partition function throughout~$R_\Delta$.
More recently, correlation decay techniques have been extended to obtain
deterministic approximation algorithms for the anti-ferromagnetic Ising partition function
on hypergraphs over a range of parameters~\cite{lu16:_fptas_hardc}, as well as to
several other problems not related to the Ising model.
In the ferromagnetic setting, however, for reasons mentioned earlier,
correlation decay techniques have had more limited success:
Zhang, Liang and Bai~\cite{zhaliabai09} handle
only the ``high-temperature'' regime of the Ising model, while the
recent results for ferromagnetic two-spin systems of Guo and Lu~\cite{guo2016uniqueness}
do not apply to the case of the Ising model.

In a parallel line of work, Barvinok initiated the study of Taylor
approximation of the logarithm of the partition function, which led to
quasipolynomial time approximation algorithms for a variety of counting problems\cite{barvinok_computing_2015,Barvinok15,BarvinokSoberon16a,BarvinokSoberon16b}.
More recently, Patel and Regts~\cite{patel_deterministic_2016} showed
that for several models that can be written as induced subgraph sums, one
can actually obtain an FPTAS from this approach.  In particular, for problems such as counting
independent sets with negative (or, more generally, complex valued) activities on bounded degree graphs,
they were able to match the range of applicability of existing
algorithms based on correlation decay, and were also able to extend
the approach to Tutte polynomials and edge-coloring models (also known
as Holant problems) where little is known about correlation decay.

The Lee-Yang program was initiated by Lee and Yang~\cite{leeyan52} in
connection with the analysis of phase transitions.  By proving the
famous Lee-Yang circle theorem for the ferromagnetic Ising model~\cite{leeyan52b},
they were able to conclude that
there can be at most one phase transition for the model.
Asano~\cite{asano_lee-yang_1970} extended the Lee-Yang theorem to the
Heisenberg model, and provided a simpler proof.
Asano's work was generalized further by Suzuki and Fisher~\cite{suzuki1971zeros},
while Sinclair and Srivastava~\cite{sinclair12:_lee_yang} studied the multiplicity
of Lee-Yang zeros.
A complete characterization of Lee-Yang polynomials that are independent
of the ``temperature'' of the model was recently obtained by
Ruelle~\cite{ruelle_characterization_2010}.  The study of Lee-Yang
type theorems for other statistical physics models has also generated
beautiful connections with other areas of mathematics.  For example,
Shearer~\cite{Shearer} and Scott and Sokal~\cite{ScottSokal}
established the close connection between the location of the zeros of the
independence polynomial and the Lov\'asz Local Lemma, while the study
of the zeros of generalizations of the matching polynomial was used in
the recent celebrated work of Marcus, Spielman and Srivastava on the
existence of Ramanujan graphs~\cite{marcus_interlacing_2015}. Such
Lee-Yang type theorems are exemplars of the more general stability theory
of polynomials~\cite{borcea2009lee1,borcea2009lee2}, a field of study
that has had numerous recent applications to theoretical computer
science and combinatorics~(see, e.g.,
\cite{borcea2009negative,marcus_interlacing_2015,marcus2015interlacing,anari2014kadison,anari_generalization_2017,sinclair12:_lee_yang,straszak_real_2016}).

\section{Approximation of the log-partition function by Taylor series}
\label{sec:taylor}
In this section we present an approach due to
Barvinok~\cite{barvinok_computing_2015} for approximating the partition
function of a physical system by truncating the Taylor series of its logarithm, as
discussed in the introduction.  We will work in our most general setting of
symmetric two-spin systems on hypergraphs, which of course includes the
Ising model (on graphs or hypergraphs) as a special case.
As in~(\ref{eqn:hypergraph2spin}), such a system has partition function
\begin{displaymath}
  Z_H^{\bm \varphi}(\lambda) =\sum_{\bm \sigma: V \to \set{+,-}} \prod_{ e \in E} \varphi_e\inp{\bm \sigma\big|_e} \lambda^{\abs{\set{v:\bm \sigma(v) = +}}}.
\end{displaymath}
Our goal is a FPTAS for $Z_H^{\bm \varphi}(\lambda)$, i.e., a deterministic algorithm that, given
as input $H$, $\{\varphi_e\}$, $\lambda$ with $|\lambda|\ne 1$ and $\varepsilon\in(0,1]$,
runs in time polynomial
in $n=|H|$ and $\varepsilon^{-1}$ and outputs a $(1\pm \varepsilon)$-multiplicative approximation
of~$Z_H^{\bm \varphi}(\lambda)$, i.e., a number~$\hat Z$ satisfying
\begin{equation}\label{eq:fptasdef}
  |\hat Z - Z_H^{\bm \varphi}(\lambda)| \le \varepsilon |Z_H^{\bm \varphi}(\lambda)|.
\end{equation}
(Note that in our setting $\hat Z$ and $Z_H^{\bm \varphi}(\lambda)$
may be complex numbers.)  By the symmetry
$\varphi_e(\bm \sigma) = \overline{\varphi_e(-\bm \sigma)}$, we also
have
$Z^{\bm \varphi}(\lambda) = \lambda^n Z^{{\overline{\bm
      \varphi}}}(\frac{1}{\lambda})$, so that without loss of
generality we may assume $\abs{\lambda} < 1$.

For fixed $H$ and (hyper)edge activities $\bm \varphi$, we will write
$Z(\lambda) = Z_H^{\bm \varphi}(\lambda)$ for short.
Letting $f(\lambda) = \log Z(\lambda)$, using the Taylor expansion around $\lambda=0$ we get
\begin{equation}
  f(\lambda) = \sum_{j=0}^{\infty} f^{(j)}(0) \cdot
  \frac{\lambda^j}{j!}, \label{eq:12}%
\end{equation}
where $f(0) = \log Z(0) = 0$.
Note that $Z=\exp(f)$, and thus an additive error in $f$ translates to a multiplicative error in $Z$.
More precisely, given $\eps\le 1/4$, and $f, \widetilde{f}\in \C$ such that $|f - \widetilde{f}| \le \eps$, we have
\[
	|\exp(\widetilde{f})- \exp(f) | = |\exp(\widetilde{f} - f) - 1|\times |\exp(f)| \le
	4\eps |\exp(f)|,
\]
where the last inequality, valid for $\varepsilon\le 1/4$, follows by elementary complex analysis.
In other words, to achieve a multiplicative approximation of~$Z$ within a factor $1\pm \varepsilon$,
as required by a FPTAS, it suffices to obtain an additive approximation of~$f$ within~$\varepsilon/4$.

To get an additive approximation of $f$, we use the first $m$ terms in the Taylor expansion.
Specifically, we compute $f_m(\lambda) := \sum_{j=0}^{m} f^{(j)}(0) \cdot  \frac{\lambda^j}{j!}$.
We show next how to compute the derivatives $f^{(j)}(0)$ from the derivatives of~$Z$ itself
(which are more readily accessible).

To compute $f^{(j)}(0)$, note that $f'(\lambda) = \frac{1}{Z(\lambda)} \D{Z(\lambda)}{\lambda}$,
or $\D{Z(\lambda)}{\lambda} = f'(\lambda) Z(\lambda)$. Thus for any $m\ge 1$,
\Style{DDisplayFunc=outset,DShorten=true}
\begin{align}
	\D[m]{Z(\lambda)}{\lambda} = \sum_{j=0}^{m-1} {m-1 \choose j} \D[j]{Z(\lambda)}{\lambda} \cdot \frac{\mathrm{d}^{m-j}}{\mathrm{d}\lambda^{m-j}}f(\lambda). %
	\label{eqn:deriv-coeff}
\end{align}
Given $\D[j]{Z(\lambda)}{\lambda} \big|_{\lambda=0}$ for
$j=0, \cdots, m$, eq.\nobreakspace \textup {(\ref {eqn:deriv-coeff})} is a triangular system of linear equations in
$\set{f^{(j)}(0)}_{j=1}^{m}$ of representation length $\poly{m}$, and is non-degenerate
since $Z(0)=1$; hence it can be solved in $\poly{m}$ time.

We can now specify the algorithm: first compute $\set{\D[j]{Z(\lambda)}{\lambda} \big|_{\lambda=0}}_{j=0}^m$; next, use the system in~eq.\nobreakspace \textup {(\ref {eqn:deriv-coeff})} to solve for $\set{f^{(j)}(0)}_{j=1}^{m}$; and finally, compute and ouput the approximation $f_m(\lambda)$.

To quantify the approximation error in this algorithm, we need to
study the locations of the complex roots $r_1, \cdots, r_n$ of~$Z$.
Throughout this paper, we will be using (some variant of) the Lee-Yang
theorem to argue that, for the range of interactions~$\bm\varphi$ we
are interested in, the roots $r_i$ all lie on the unit circle in the
complex plane, i.e., $|r_i| = 1$ for all~$i$.  Note that since we are
assuming that $\varphi_e(-,\cdots, -) = 1$, the constant term
$\prod_{i=1}^n(-r_i)$ of $Z(\lambda)$ is $1$, and hence we have
$Z(\lambda)=\prod_i (1 - \frac{\lambda}{r_i})$.  The log partition
function can then be written as
\begin{equation}
  f(\lambda) = \log Z(\lambda) = \sum_{i=1}^n \log \inp{1 - \frac{\lambda}{r_i}}
  = -\sum_{i=1}^n \sum_{j=1}^\infty \frac{1}{j} \inp{\frac{\lambda}{r_i}}^j.\label{eq:13}
\end{equation}
Note that due to the uniqueness of the Taylor expansion of
meromorphic functions, the two power series expansions of $f(\lambda)$
in eqs.\nobreakspace \textup {(\ref {eq:12})} and\nobreakspace  \textup {(\ref {eq:13})} are identical in the domain of their
convergence. Denoting the first $m$ terms of the above expansion by
$f_m(\lambda) = -\sum_{i=1}^n \sum_{j=1}^m \frac{1}{j}
(\frac{\lambda}{r_i})^j$, the error due to truncation is bounded
by
\[
	\abs{f(\lambda) - f_m(\lambda)} \le n \sum_{j=m+1}^\infty \frac{\abs{\lambda}^j}{j} \le \frac{n\abs{\lambda}^{m+1}}{(m+1)(1-\abs{\lambda})},
\]
recalling that by symmetry we are assuming $\abs{\lambda} <1$.
Thus to get within $\eps/4$ additive error, it suffices to take 
$m \ge \frac{1}{\log(1/\abs{\lambda})}\bigl({\log (\frac{4n}{\varepsilon}) +
    \log(\frac{1}{1-\abs{\lambda}})}\bigr)$.
The following result summarizes our discussion so far.
\begin{lemma}
  Given $\varepsilon \in (0, 1)$,
  $m \ge \frac{1}{\log(1/\abs{\lambda})}\bigl({\log (\frac{4n}{\varepsilon}) +
    \log(\frac{1}{1-\abs{\lambda}})}\bigr)$, and the values
  $\set{\D[j]{Z(\lambda)}{\lambda} \big|_{\lambda=0}}_{j=0}^m$,
	$f_m(\lambda)$ can be computed in time $\poly{n/\eps}$. Moreover, if the Lee-Yang theorem
	holds for the partition function $Z(\lambda)$, then $\abs{f_m(\lambda) - f(\lambda)} < \eps/4$, and thus $\ \exp(f_m(\lambda))$ approximates $Z(\lambda)$ within a multiplicative factor $1\pm\varepsilon$.
	\label{lem:taylor-approx}
\end{lemma}

The missing ingredient in turning Lemma~\ref{lem:taylor-approx} into an FPTAS is the
computation of the derivatives 
$\D[j]{Z(\lambda)}{\lambda} \big|_{\lambda=0}$ for $1\le j\le m$, 
which themselves are just multiples of the $m+1$ lowest-degree coefficients of~$Z$.  Computing
these values naively using the definition of $Z(\lambda)$ requires $n^{\Omega(m)}$ time.
Since $m$ is required to be of order $\Omega(\log(n/\varepsilon))$, this results in only a
quasi-polynomial time algorithm.
In the next section, we show how to compute these values in polynomial time when
$H$ is a hypergraph of bounded degree and bounded hyperedge size, which
when combined with Lemma~\ref{lem:taylor-approx} gives an FPTAS.

\section{Computing coefficients via \higraphs{}}\label{sec:coefficients}
As discussed in the introduction, Patel and Regts~\cite{patel_deterministic_2016} recently
introduced a technique for efficiently computing the low-degree coefficients of a partition function using
induced subgraph counts.
In this section we introduce the notion of \emph{\sub\higraph{} counts},
and show how it allows the Patel-Regts framework to be adapted to any
hypergraph two-spin system with vertex activities (including the Ising
model with vertex activities as a special case).
We will align our notation with~\cite{patel_deterministic_2016} as much as possible.
From now on, unless otherwise stated, we will use $G$ to denote a hypergraph.
Recall from the introduction the partition function of a two-spin system on a hypergraph $G=(V,E)$:
\begin{align}
  Z_G^{\bm \varphi}(\lambda) =\sum_{\bm \sigma: V \to \set{+,-}} \prod_{ e \in E} \varphi_e\inp{\bm \sigma\big|_e} \lambda^{\abs{\set{v:\bm\sigma(v) = +}}}.
  \label{eqn:hyperZ}
\end{align}
Due to the normalization $\varphi_e(-,\cdots, -) = 1$, each term in the summation depends only on the set $S = \set{v:\bm\sigma(v) = +}$ and the labelled induced sub-hypergraph 
$\inp{S \cup \partial S, E[S]\cup E(S, \overline{S}) }$, 
where $E[S]$ is the set of edges within~$S$,
$\partial S$ is the boundary of $S$ defined as $\partial S \defeq \bigcup_{v \in S} N_G(v) \setminus S$, and $N_G(v)$ is the set of vertices adjacent to the vertex $v$ in~$G$.
This fact motivates the induced sub-structures we will consider.

Let $\bm \sigma^S$ be the configuration where the set of vertices assigned $+$-spins is $S$, that is, $\bm\sigma^S(v) = +$ for $v \in S$ and $\bm\sigma^S(v)=-$ otherwise.
We will also write $\varphi_e(S) \defeq \varphi_e(\bm \sigma^S\big|_e)$ for simplicity.
Thus the partition function can be written
\begin{displaymath}
  Z_G^{\bm \varphi}(\lambda) = \sum_{S \subseteq V } \prod_{e: e \cap S \neq \emptyset} \varphi_e(S) \lambda^{\abs{S}}.
\end{displaymath}

We start with the standard factorization of the partition function in
terms of its complex zeros $r_1,\ldots,r_n$, where $n=|V|$.  As
explained in the paragraph preceding eq.~\eqref{eq:13}, the assumption
$\varphi_e(-,\cdots, -) = 1$ allows one to write the partition
function as
\begin{align*}
  Z_G^{\bm \varphi}(\lambda) = \prod_{j=1}^n (1 - \lambda/r_j) = \sum_{i=0}^n (-1)^i e_i(G) \lambda^i,
\end{align*}
where $e_i(G)$ is the elementary symmetric polynomial of degree $i$
evaluated at $(\frac{1}{r_1}, \cdots, \frac{1}{r_n})$.

On the other hand, we can also express the coefficients $e_i(G)$
combinatorially using the definition of the partition function:
\begin{align}
  e_i(G) = (-1)^i \sum_{\substack{S \subseteq V\\ \abs{S}=i }}\;
  \prod_{e: e \cap S \neq \emptyset} \varphi_e(S).
  \label{eq:eiGdef}
\end{align}

Once we have computed the first $m$ coefficients of $Z$ (i.e., the
values $e_i(G)$ for $i=1,\cdots,m$), where
$m=\Omega\inp{\frac{\log (n/\eps) -
    \log(1-\abs{\lambda})}{\log(1/\abs{\lambda})}}$, we can use
\MakeUppercase Lemma\nobreakspace \ref {lem:taylor-approx} to obtain an FPTAS as claimed in
Theorems~\ref{thm:main1}, \ref{thm:main2} and~\ref{thm:main3}.
The main result in this section will be an algorithm for computing these coefficients $e_i(G)$:
\begin{theorem}
Fix $k,\Delta\in\N$ and $C >0$.  There exists a deterministic $\poly{n/\eps}$-time algorithm
that, given any $n$-vertex hypergraph~$G$ of maximum degree $\Delta$ and maximum
hyperedge size~$k$, and any $\eps \in (0,1)$, computes the coefficients $e_i(G)$ for $i=1,\cdots,m$, where $m=\lceil C \log(n/\eps) \rceil$.
	\label{thm:elem-sym0}
	\label{thm:elem-sym0}
\end{theorem}

\subsection{\Higraphs{} in a hypergraph}
\label{sec:higraphs-hypergraph}
To take advantage of the fact that each term in~eq.\nobreakspace \textup {(\ref {eqn:hyperZ})} only
depends on the set $S$ and the induced sub-hypergraph
$\inp{S \cup \partial S, E[S]\cup E(S, \overline{S}) }$, we define the
following structure.
\begin{definition}%
  Given a vertex set $S$ and a set $E$ of hyperedges, $H = (S, E)$ is
  called \a\ \emph{\higraph{}} if for all $e \in E$,
  $e \cap S \neq \emptyset$.  The set $S$ is called the \emph{label
    set} of the \higraph{} $H$ and the set
  $B(H) \defeq \inp{\bigcup_{e \in E}e} \setminus S$ is called the
  \emph{boundary set}.%
\end{definition}
Given \a\ \higraph{} $H$, we use the notation
$\underline{V}(H)$ for its label set.  The \emph{size} $|H|$ of the
\higraph{} $H$ is defined to be $\abs{\underline{V}(H)}$.  \A\
\higraph{} $H = (S, E)$ is said to be \emph{connected} if the
hypergraph $(S, \inbr{e \cap S \st e \in E})$ is connected.  It is
said to be \emph{disconnected} otherwise. In the latter case, there
exists a partition of $S$ into non-empty sets $S_1, S_2$, and a
partition of $E$ into sets $E_1$ and $E_2$, such that $(S_i, E_i)$ are
\higraphs{} for $i = 1, 2$, and the sets $S_2 \cap B(H_1)$ and
$S_1 \cap B(H_2)$ are empty.  In this case, we write
$H = H_1\uplus H_2$, and say that the \higraphs{} $H_1$ and $H_2$ are
\emph{disjoint}.  
(Note that disjoint \higraphs{} may share boundary vertices.)

\begin{remark}%
  Note that a hypergraph $G = (V, E)$ can itself be viewed as \a\
  \higraph{}.  However, as is clear from the definition, not all
  \higraphs{} are hypergraphs.
\end{remark}

In order to exploit the structure of the terms in~eq.\nobreakspace \textup {(\ref {eqn:hyperZ})}
alluded to above, we now define the notion of an \emph{induced
  \sub\higraph{}} of \a\ \higraph{}.  Given \a\ \higraph{}
$H = (S, E)$ and a subset $S'$ of $S$, we define the induced
\sub\higraph{} $\ind{H}{S'}$ as
$(S', \inbr{e \in E \st e \cap S' \neq \emptyset})$.
Further, we say that \a\ \higraph{}
$H$ is an induced \sub\higraph{} of \a\ \higraph{} $G$, denoted
\subinsect{H}{G}, if there is a set $S \subseteq \underline{V}(G)$
such that $\ind{G}{S} = H$.  

\subsection{Weighted \sub\higraph{} counts}
Just as graph invariants may be expressed as sums over induced
subgraph counts, we will consider weighted \sub\higraph{} counts of
the form $f(G) = \sum_{S\subseteq \underline{V}(G)} a_{\ind{G}{S}}$
and the functions~$f$ expressible in this way.  
Here $G$ is any insect,
and the coefficients $a_H$ depend only on~$H$, not on~$G$.

Let $\Gk$ be the set of \higraphs{} up to size $t$, with maximum
degree $\Delta$ and maximum hyperedge size $k$.  Note that since
\higraphs{} are labelled, this is an infinite set.  We will fix
$\Delta$ and $k$ throughout, and write $\G \defeq \bigcup_{t\ge 1} \Gk$.
Let $\isubinsect{H}{G}$ be the indicator that $H$ is an induced \sub\higraph{}
of $G$, that is,
\[
  \isubinsect{H}{G} = 1 \text{ if there is a set
    $S \subseteq \underline{V}(G)$ such that $\ind{G}{S} = H$, and $0$
    otherwise}.
\]
A weighted \sub\higraph{} count $f(G)$ of the form considered above
can then also be written as
$ f(G) = \sum_{H \in \G} a_{H} \cdot \isubinsect{H}{G}$.  This
alternative notation helps simplify the presentation of some of the
combinatorial arguments below.  Note that even though $\mathcal{G}$ is
infinite, the above sum has only finitely many non-zero terms for any
finite \higraph{} $G$.  Further, as \higraphs{} are labelled, $f(G)$
may also depend on the labelling of $G$, unlike a graph invariant
where isomorphic copies of a graph yield the same value.

\label{sec:prop-induc-sub}
A weighted \sub\higraph{} count $f$ is said to be \emph{additive} if,
given any two disjoint \higraphs{} $G_1$ and $G_2$,
$f(G_1 \uplus G_2) = f(G_1) + f(G_2)$.  An argument due to
Csikv\'{a}ri and Frenkel~\cite{csikvari2016benjamini}, also employed
in the case of graph invariants by Patel and
Regts~\cite{patel_deterministic_2016}, can then be adapted to give the
following:
\begin{lemma}
  Let $f$ be a weighted \sub\higraph{} count, so that $f$ may be
  written as
  \[
    f(G) = \sum_{H \in \G} a_{H}\cdot \isubinsect{H}{G}.
  \]
  Then $f$ is additive if and only if $a_{H} = 0$ for
  all \higraphs{} $H$ that are disconnected.
  \label{lem:additive}
\end{lemma}

\begin{proof}
  When $H$ is connected, we have
  $\isubinsect{H}{G_1 \uplus G_2} = \isubinsect{H}{G_1} + \isubinsect{H}{G_2}$; thus $f$ given
  in the above form is additive if $a_{H'} = 0$ for all $H'$ that
  are not connected.

 Conversely, suppose $f$ is additive.  By the last paragraph, we can assume
  without loss of generality that the sequence $a_H$ is supported
  on disconnected \higraphs{} (by subtracting the component of $f$
  supported on connected $H$).  We now show that for such an $f$,
  $a_{H}$ must be $0$ for all disconnected $H$ as well.

  For if not, 
let $H$ be a (necessarily disconnected) \higraph{} of smallest size
  for which $a_{H} \neq 0$.
  Since $a_J=0$ for all insects~$J$ with $|J|<|H|$,
  we must have $f(J)=0$ for all such insects.
Also, since $H$ is disconnected, there exist non-empty \higraphs{} $H_1$ and
  $H_2$ such that $H = H_1 \uplus H_2$.  By additivity, we then have
  $f(H) = f(H_1) + f(H_2) = 0$, where the last equality follows since
  both $|H_1|,|H_2|$ are strictly smaller than $|H|$.  On the other
  hand, since $H$ is \a\ \higraph{} with the smallest possible number
  of vertices such that $a_{H} \neq 0$,
  we also have
  $f(H) = a_{H}\isubinsect{H}{H} = a_H$.  This implies $a_{H} = 0$, which
  is a contradiction.  Hence we must have $a_H=0$ for all disconnected~$H$.
\end{proof}

The next lemma implies that the product of weighted \sub\higraph{} counts
can also be expressed as a weighted \sub\higraph{} count.
We begin with a definition.
\begin{definition}
  \A\ \higraph{} $H_1=(S_1,E_1)$ is \emph{compatible} with another
  \higraph{} $H_2=(S_2,E_2)$ if the \higraph{}
  $H \defeq (S_1 \cup S_2, E_1 \cup E_2)$ satisfies
  $\ind{H}{S_1} = H_1$ and $\ind{H}{S_2} = H_2$.
\end{definition}

\begin{lemma}
  \label{lem:product}
  Let $H_1=(S_1, E_1), H_2=(S_2,E_2)$ be arbitrary \higraphs{}.
  \begin{itemize}
  \item[(i)] If $H_1$ and $H_2$ are not compatible, then there is no
    \higraph{} $G$ such that $\subinsect{H_1}{G}$ and
    $\subinsect{H_2}{G}$.  In other words, for every \higraph{} $G$,
    \[
      \isubinsect{H_1}{G} \isubinsect{H_2}{G} = 0.
    \]
  \item[(ii)] If $H_1$ and $H_2$ are compatible, then for every \higraph{}
    $G$,
  \begin{align*}
    \isubinsect{H_1}{G} \isubinsect{H_2}{G} =  \isubinsect{H}{G},
  \end{align*}
  where $H$ is the \higraph{} $(S_1 \cup S_2, E_1 \cup E_2)$, and
  satisfies $\ind{H}{S_i} = H_i$ for $i = 1, 2$.
\end{itemize}
\end{lemma}

\begin{proof}
  We start by making two observations. First, if $\ind{G}{S_1} = H_1$
  and $\ind{G}{S_2} = H_2$ then
  $\ind{G}{S_1 \cup S_2} = H = (S_1 \cup S_2, E_1 \cup E_2)$. Second,
  if $T \subseteq S \subseteq V$ and $G_1 \defeq \ind{G}{S}$ then
  $\ind{G_1}{T} = \ind{G}{T}$.

  Suppose first that $H_1$ and $H_2$ are not compatible.  Suppose, for the
  sake of contradiction, that there exists an insect~$G$ such that
   $\ind{G}{S_i} = H_i$ for $i = 1,2$.  Then, from the
  first observation above we have
  $\ind{G}{S_1 \cup S_2} = H = (S_1 \cup S_2, E_1 \cup E_2)$, while
  from the second observation we have
  $\ind{H}{S_i} = \ind{G}{S_i} = H_i$ for $i = 1,2$.  This contradicts
  the assumption that $H_1$ and $H_2$ are incompatible.  Thus, we must
  have $\isubinsect{H_1}{G} \isubinsect{H_2}{G} = 0$ for every $G$, 
  proving part~(i).

  Now suppose that $H_1$ and $H_2$ are compatible. As seen above,
  $\ind{G}{S_i} = H_i$ for $i = 1,2$ implies that
  $\ind{G}{S_1 \cup S_2} = H$.  On the other hand, if
  $\ind{G}{S_1 \cup S_2} = H$, then by the compatibility of $H_1$ and
  $H_2$, and the second observation above,
  $\ind{G}{S_i} = \ind{H}{S_i} = H_i$ for $i = 1, 2$. This proves part~(ii)
of the lemma.  %
\end{proof}

An immediate corollary of the above lemma is that a product of
weighted \sub\higraph{} counts is also a \sub\higraph{} count
supported on slightly larger \higraphs{}.

\begin{corollary}
  \label{cor:weigth-count}
  If $f_1(G) = \sum_{H}a_H\cdot\isubinsect{H}{G}$ and
  $f_2(G) = \sum_{H} b_H\cdot\isubinsect{H}{G}$ are weighted
  \sub\higraph{} counts, then so is $g(G) \defeq f_1(G)f_2(G)$.
  Moreover, if $f_1,f_2$ are supported on \sub\higraphs{} of sizes
  $\le t_1, t_2$ respectively (i.e., if $a_H = 0$ when $\abs{H} > t_1$
  and $b_H = 0$ when $\abs{H} > t_2$), then $g$ is supported on
  \sub\higraphs{} of size $\le t_1+t_2$.
\end{corollary}

\begin{proof}
  For compatible \higraphs{} $H_i = (S_i, E_i)$ we denote by
  $H_1 \cup H_2$ the \higraph{} $(S_1 \cup S_2, E_1 \cup E_2)$. Now,
  for any \higraph{} $G$ we have,
  \begin{align}
    g(G)
    &= \sum_{H_1, H_2} a_{H_1} b_{H_2}
      \cdot\isubinsect{H_1}{G}\cdot\isubinsect{H_2}{G}\nonumber\\
    &= \sum_{H_1, H_2 \text{ compatible}} a_{H_1}b_{H_2}
      \cdot \isubinsect{H_1 \cup H_2}{G}\nonumber\\[5pt]
    &= \;\,\sum_{H} c_H \cdot \isubinsect{H}{G},\nonumber\\
\noalign{\hbox{\rm where in the second line we have used Lemma~\ref{lem:product}, and where}}
    c_H &\defeq %
    \sum_{\substack{H_1, H_2 \text{ compatible}\\H = H_1 \cup H_2}}%
    a_{H_1} b_{H_2}.
    \label{eq:14}
 \end{align}
  Note that the number of non-zero terms in the definition of each
  $c_H$ is finite, and that $|H_1\cup H_2| \le |H_1|+|H_2|$. This completes the proof.
\end{proof}

\subsection{Enumerating connected \sub\higraphs{}}
We observe next that $e_i(G)$, as defined in~eq.\nobreakspace \textup {(\ref {eq:eiGdef})},
can be written as a weighted \sub\higraph{} count.
Accordingly, we generalize eq.\nobreakspace \textup {(\ref {eq:eiGdef})} to arbitrary
insects~$G$ of maximum degree~$\Delta$ and hyperedge size~$k$
as follows:
\begin{align}
  e_i(G) = (-1)^i \sum_{\substack{S \subseteq \underline{V}(G)\\ \abs{S}=i }}\,
  \prod_{e: e \cap S \neq \emptyset} \varphi_e(S) = \sum_{H \in \Gi} \mu_{H} \cdot \isubinsect{H}{G},
	\label{eqn:elem-sym}
\end{align}
where
$\mu_{H} := (-1)^i \prod_{e: e \cap \underline{V}(H) \neq \emptyset}
\varphi_e(\underline{V}(H))$.
Note that this definition coincides with eq.\nobreakspace \textup {(\ref {eq:eiGdef})} when $G$
is a hypergraph, and also extends the definition of the partition
function from hypergraphs to \higraphs{} via the equation
$Z_G(\lambda) = \sum_{i=0}^{|G|}(-1)^ie_i(G)\lambda^i$; when
$G = (S, E)$ this definition is equivalent to that of the partition
function on the hypergraph $(S \cup B(G), E)$, with the vertices in
$B(G)$ set to `$-$'.  This latter observation implies that when the
\higraph{} $G$ is disconnected and $G = G_1 \uplus G_2$, we have
$Z_{G}(\lambda) = Z_{G_1}(\lambda) Z_{G_2}(\lambda)$.

We now consider the computational properties of the above expansion.
Note that each coefficient $\mu_{H}$ is readily computable in time
$\poly{\abs{H}}$; however, as discussed in the introduction, the
number of $H \in \Gi$ such that $\isubinsect{H}{G} \neq 0$ is
$\Omega(n^i)$, so that a naive computation of $e_i(G)$ using
eq.\nobreakspace \textup {(\ref {eqn:elem-sym})} would be inefficient.
To prove Theorem\nobreakspace \ref {thm:elem-sym0}, we consider the set of \emph{connected}
\higraphs{}, denoted by $\Ci$, rather than $\Gi$.  We will show in this
subsection that $\Ci$ can be efficiently enumerated, and then in the
following subsection reduce the above summation over $\Gi$ to
enumerations of $\Ci$.
As in~\cite{patel_deterministic_2016},
we use the following calculation of Borgs et al.~\cite[Lemma
2.1 (c)]{borgs_left_2013}.
\begin{lemma}
	Let $G$ be a multigraph with maximum degree $\Delta$ (counting edge multiplicity) 
	and let $v$ be a vertex of~$G$.
	Then the number of subtrees of $G$ with $t$ vertices containing the vertex $v$ is at most $\frac{(e \Delta)^{t-1}}{2}$.
	\label{lem:ind-subgraph-cnt}
\end{lemma}
\begin{proof}
  Consider the infinite rooted $\Delta$-ary tree $T_\Delta$.  The number of subtrees with $t$ vertices starting from the root is $\frac{1}{t}{t \Delta \choose t-1} <\frac{(e \Delta)^{t-1}}{2} $. (See also~\cite[Theorem 5.3.10]{stanley1999enumerative}.)  The proof is completed by observing 
 that the set of $t$-vertex subtrees of~$G$ containing vertex~$v$ can be mapped injectively
 into subtrees of $T_\Delta$ containing the root.
\end{proof}
\begin{corollary}
  Let $G$ be a hypergraph with maximum degree $\Delta$ and maximum
  hyperedge size $k$, and let $v\in V(G)$.  Then the number of
  connected induced \sub\higraphs{} of $G$ of size $t$ whose label set
  contains the vertex $v$ is at most $\frac{(e \Delta k)^{t-1}}{2}$.
	\label{cor:ind-substruct-cnt}
\end{corollary}
\begin{proof}
  Consider the multigraph $H$ obtained by replacing every hyperedge of size $r$ in~$G$ by an $r$-clique.
  For any connected induced \sub\higraph{} $A$ of $G$, the label set
  $\underline{V}(A)$ is connected in $H$.
  Now, for any two distinct connected induced \sub\higraphs{} $A$ and
  $B$, let $S_A$ and $S_B$ be the sets of trees in $H$ that span the
  label sets $\underline{V}(A)$ and $\underline{V}(B)$ of $A$ and $B$
  respectively. %
  Since the label sets of $A$ and $B$ are different, we must have
  $S_A \cap S_B = \emptyset$.
  Thus the number of connected subtrees on $t$ vertices in $H$ which
  contain the vertex $v$ is an upper bound on the number of connected
  induced \sub\higraphs{} in $G$ whose label set contains $v$.

  Finally, in the multigraph $H$ the maximum degree is $\Delta k$, so by
  \MakeUppercase Lemma\nobreakspace \ref {lem:ind-subgraph-cnt} the number of such subtrees is at most
  $\frac{(e \Delta k)^{t-1}}{2}$.  
\end{proof}

As a consequence we can efficiently enumerate all connected induced \sub\higraphs{} of logarithmic size in a bounded degree graph.
This follows from a similar reduction to a multigraph, applying~\cite[Lemma 3.4]{patel_deterministic_2016}. However, for the sake of completeness we also include a direct proof.
\begin{lemma}
  For a hypergraph $G$ of maximum degree $\Delta$ and maximum
  hyperedge size $k$, there exists an algorithm that enumerates
  all connected induced \sub\higraphs{} of size at most $t$ in~$G$ and runs in time
  $\Ot(n t^3 (e\Delta k)^{t+1})$. Here $\Ot$ hides factors of the form
  $\polylog{n}, \polylog{\Delta k}$ and $\polylog{t}$.
	\label{lem:enumerator}
\end{lemma}
\begin{proof}
	Let $\T_t$ be the set of $S\subseteq V(G)$ such that $\abs{S}\le t$ and $G^+[S]$ is connected.
	Note that given $S\in\T_t$, $G^+[S]$ will be a \sub\higraph{} of size $t$, and this clearly enumerates all of them.
	Also, by \MakeUppercase Corollary\nobreakspace \ref {cor:ind-substruct-cnt}, $\abs{\T_t} \le O(n (e \Delta k)^{t})$.
	Thus it remains to give an algorithm to construct $\T_t$ in about the same amount of time.

	We construct $\T_t$ inductively. For $t=1$, $\T_1\defeq V(G)$.
	Then given $\T_{t-1}$, define the multiset
	\[
		\T^*_t \defeq \T_{t-1} \cup \set{ S \cup \set{v} : S \in \T_{t-1} \ \textrm{ and } \ v \in N_G(S)\setminus S}.
	\]
	Since $\abs{N_G(S)} < t \Delta k$, we can compute the set $N_G(S) \setminus S$ 
	in time $O(t \Delta k)$, and construct $\T^*_t$ in time $\tilde{O}(\abs{\T_{t-1}} t^2\Delta k)=\tilde{O}(n t^2 (e\Delta k)^{t})$.
	Finally, we remove duplicates in $\T^*_t$ to get $\T_t$ (e.g., by sorting the sets $S \in \T^*_t$, where each is represented as a string of length $\Ot(t)$), in time $\Ot(nt^3(e\Delta k)^{t})$.

	Starting from $\T_1$, inductively we perform $t$ iterations to construct $\T_t$.  Thus the overall running time is $\sum_{i=1}^t \Ot(n i^3 (e\Delta k)^{i}) = \Ot(n t^3 (e\Delta k)^{t+1})$.
\end{proof}

\subsection{Proof of Theorem\nobreakspace \ref {thm:elem-sym0}}
The results in the previous subsection allow us to efficiently
enumerate connected \sub\higraphs{}.  To prove~\MakeUppercase Theorem\nobreakspace \ref {thm:elem-sym0},
it remains to reduce the sum over all (possibly disconnected)~$H$
in~eq.\nobreakspace \textup {(\ref {eqn:elem-sym})} to a sum over {\it connected}~$H$.  We now show
that the method of doing so using Newton's identities and the
multiplicativity of the partition function developed by Patel and
Regts~\cite{patel_deterministic_2016} for graphs extends to the case
of \higraphs{}.
Let $G$ be any \higraph{} of size $n$ and consider the $t$-th power
sum of the inverses of the roots $r_i$, $1 \leq i \leq n$, of
$Z_G(\lambda)$ (extended to insects~$G$ as in the paragraph following eq.\nobreakspace \textup {(\ref {eqn:elem-sym})}):
\[
  p_t(G) = \sum_{i=1}^n \frac{1}{r_i^t}.
\]
Now by Newton's identities (which relate power sums to elementary
symmetric polynomials), we have
\begin{align}
  p_t = \sum_{i=1}^{t-1} (-1)^{i-1} p_{t-i} e_i + (-1)^{t-1} t e_t.
\label{eqn:newton}
\end{align}
Recall from~eq.\nobreakspace \textup {(\ref {eqn:elem-sym})} that $e_i$ is a weighted
\sub\higraph{} count supported on insects of size $\le i$, 
and also from Corollary~\ref{cor:weigth-count}
that the product of two weighted \sub\higraph{} counts supported on insects of
size $\le t_i,t_2$ respectively is a weighted \sub\higraph{} count supported
on insects of size $\le t_1+t_2$. Therefore, by eq.\nobreakspace \textup {(\ref {eqn:newton})} and induction,
each $p_t$ is also a weighted
\sub\higraph{} count supported on insects of size $\le t$.  Thus, for any \higraph{} $G$,
we may write
\begin{align}
  p_t(G) = \sum_{H\in \Gk} a^{(t)}_{H} \cdot \isubinsect{H}{G}
\label{eqn:power-sum-dontuseme}
\end{align}
for some coefficients $a^{(t)}_{H}$ to be determined.  (The
superscript $(t)$ reflects the fact that a given~$H$ will in general
have different coefficients for different~$p_t$.)

Recall now that if $G$ is disconnected with $G = G_1 \uplus G_2$ then
$Z_{G}(\lambda) = Z_{G_1}(\lambda) \cdot Z_{G_2}(\lambda)$.  Thus, the
polynomials $Z_G(\lambda)$ are multiplicative over $G$, and hence sums
of powers of their roots, such as $p_t(G)$ are additive:
$p_t(G_1 \uplus G_2) = p_t(G_1) + p_t(G_2)$.  Hence by
\MakeUppercase Lemma\nobreakspace \ref {lem:additive}, the coefficients of $p_t$ are supported on {\it connected\/}
\higraphs{}, and we may write eq.\nobreakspace \textup {(\ref {eqn:power-sum-dontuseme})} as
\begin{align}
  p_t(G) = \sum_{H\in \Ck} a^{(t)}_{H} \cdot \isubinsect{H}{G}.
\label{eqn:power-sum}
\end{align}
Notice that by~\MakeUppercase Corollary\nobreakspace \ref {cor:ind-substruct-cnt}, there are at most $n (e \Delta k)^t$ non-zero terms
in this sum.

\begin{lemma}
	There is a $\poly{n/\eps}$-time algorithm to compute all the coefficients $a^{(t)}_{H}$ in eq.\nobreakspace \textup {(\ref {eqn:power-sum})}, for $t\le O(\log(n/\eps))$.
\end{lemma}
\begin{proof}
	By \MakeUppercase Lemma\nobreakspace \ref {lem:enumerator}, we compute $\T_t$, consisting of all $S\subseteq V(G)$
	 such that $\abs{S}\le t$ and $G^+[S]$ is connected.
	As we have removed duplicates, this is exactly $\Ck$.  We then use dynamic programming
	to compute the coefficients~$a^{(t)}_H$.

    By eq.~\eqref{eqn:newton}, for $t=1$ we have $p_1=e_1$, so by
    eq.~\eqref{eqn:elem-sym} we can read off the coefficients
    $a^{(1)}_{H}$ from $e_1(G)$.  Next suppose we have computed
    $a^{(t')}_{H'}$ for $\abs{H'} \le t'< t$, and we want to compute
    $a^{(t)}_{H}$ for some fixed
    connected~$H \in \mathcal{C}_t^{\Delta, k}$ such that
    $\isubinsect{H}{G}$.  Again by eq.\nobreakspace \textup {(\ref {eqn:newton})}, it suffices to
    compute the coefficient corresponding to $H$ in $p_{t-i} e_i$
    for each~$1 \leq i \leq k - 1$ (since the contribution of the last
    term in eq.~\eqref{eqn:newton} is simply $(-1)^{t-1}t\mu_H$ if
    $\abs{H} = t$ and $0$ otherwise). By eqs.~\eqref{eq:14} and
    \eqref{eqn:power-sum}, this coefficient is given by
	\begin{align}
	  \sum_{\substack{H_1\in \G_i^{\Delta, k},\, H_2 \in \mathcal{C}_{(t-i)}^{\Delta, k} \\ H_1\ \mathrm{ compatible\ with }\ H_2\\ H_1 \cup H_2 = H}} a^{(t-i)}_{H_2} \mu^{(i)}_{H_1} =
      \sum_{\substack{(S_1, S_2) \\ S_1 \cup S_2 = \underline{V}(H) \\
      H^+[S_1] \in \G_i^{\Delta, k},\, H^+[S_2] \in \mathcal{C}_{(t-i)}^{\Delta, k}}} a^{(t-i)}_{H^+[S_2]} \,\mu_{H^+[S_1]}.
		\label{eqn:hc-coeff}
	\end{align}
	Since $t\le O( \log(n/\eps))$, the second sum involves at most $4^{t} = \poly{n/\eps}$ terms.
	Moreover, due to Corollary~\ref{cor:ind-substruct-cnt}, there are at most $nt(e \Delta k)^t = \poly{n/\eps}$ previously computed $a^{(t')}_{H'}$, where $H'$ is a connected \sub\higraph{} of $G$ and $\abs{H'}\le t' <t$.
	In order to look up $a^{(t-i)}_{H^+[S]}$, one can do a linear scan, which also takes time $\poly{n/\eps}$ for $t\le O(\log(n/\eps))$.  
	The coefficients $\mu_{H^+[S]}$ can simply be read off 
	from their definition in eq.\nobreakspace \textup {(\ref {eqn:elem-sym})}.

	To conclude, because $t\le O(\log(n/\eps))$,
	eq.~\eqref{eqn:power-sum} only contains $\poly{n/\eps}$ terms.
	And for each term, $a^{(t)}_{H}$ can be computed using the above dynamic programming scheme in $\poly{n/\eps}$ time.
\end{proof}

Finally, now that we can compute $a_{H,t}$ efficiently, by~eq.\nobreakspace \textup {(\ref {eqn:power-sum})} we can compute $p_k$ using the \sub\higraph{} enumerator in~\MakeUppercase Lemma\nobreakspace \ref {lem:enumerator},
and we can then compute $e_k$ using Newton's identities as in~eq.\nobreakspace \textup {(\ref {eqn:newton})}, which
completes the proof of~\MakeUppercase Theorem\nobreakspace \ref {thm:elem-sym0}.

\subsection{Proofs of main theorems}
Our first main result in the introduction, the FPTAS for the Ising model on graphs throughout the
ferromagnetic regime with non-zero field stated in Theorem~\ref{thm:main1}, now follows
by combining~\MakeUppercase Theorem\nobreakspace \ref {thm:elem-sym0} with~\MakeUppercase Lemma\nobreakspace \ref {lem:taylor-approx} and the Lee-Yang
theorem~\cite{leeyan52b} (also stated as Theorem~\ref{thm:leeyang-graph} in the next section).
Recall from the introduction that the Lee-Yang theorem ensures that the partition function
has no zeros inside the unit disk.

Similarly, Theorem~\ref{thm:main3}, the FPTAS for two-spin systems on hypergraphs,
follows by combining~\MakeUppercase Theorem\nobreakspace \ref {thm:elem-sym0} with~\MakeUppercase Lemma\nobreakspace \ref {lem:taylor-approx} and the
Suzuki-Fisher theorem~\cite{suzuki1971zeros} (also stated as Theorem~\ref{thm:suzuki-fisher}
in the next section).  Again, the Suzuki-Fisher theorem ensures that there are no zeros
inside the unit disk, under the condition on the hyperedge activities stated in Theorem~\ref{thm:main3}.

To establish our final main algorithmic result, Theorem~\ref{thm:main2}, we first need to
prove a new Lee-Yang theorem for the hypergraph Ising model as stated in
Theorem~\ref{thm:leeyangintro} in the introduction.
This will be the content of the next and final section of our paper.
Once we have that, Theorem~\ref{thm:main2} follows immediately by the same route
as above.

\section{A Lee-Yang Theorem for Hypergraphs}\label{sec:leeyang}
In this section we prove a tight Lee-Yang theorem for the hypergraph Ising model
(Theorem~\ref{thm:leeyangintro} in the introduction).
We start by extending the definition of the hypergraph
Ising model to the multivariate setting, where each vertex and each hyperedge is allowed
to have a different activity.  As before, we have an underlying
hypergraph $G = (V, E)$ with $|V|=n$ vertices.  Given vertex
activities $\lambda_1, \lambda_2, \dots, \lambda_n$
and hyperedge activities $\vec{\beta}=(\beta_e)$, we define
\begin{displaymath}
	Z_G^{\vec{\beta}}(\lambda_1, \cdots,\lambda_n) =
	\sum_{S\subseteq V} \prod_{e \in E\inp{S, \overline{S}}} \beta_e  \prod_{i \in S} \lambda_i\,,
\end{displaymath}
where for a subset $S \subseteq V$, $E(S,\overline{S})$ is the set of hyperedges with
at least one vertex in each of~$S$ and~$\overline S$.
Note that
\begin{align}
	Z_G^{\vec{\beta}}(\lambda_1, \cdots,\lambda_n) = \prod_{i=1}^n \lambda_i \cdot Z_G^{\vec{\beta}}\inp{\frac{1}{\lambda_1}, \cdots, \frac{1}{\lambda_n}}.
	\label{eqn:multivariate-symmetry}
\end{align}

We use the following definition of the Lee-Yang property. This
definition is based on the results of Asano~\cite{asano_lee-yang_1970}
and Suzuki and Fisher~\cite{suzuki1971zeros}, and somewhat stricter
than the definition used by
Ruelle~\cite{ruelle_characterization_2010}.
\begin{definition}[\textbf{Lee-Yang property}]
  Let $P(z_1, z_2, \dots, z_n)$ be a multilinear polynomial.  $P$ is
  said to have the \emph{Lee-Yang property} (sometimes written as
  ``$P$ is LY'') if for any complex numbers
  $\lambda_1,\cdots, \lambda_n$ such that
  $\abs{\lambda_1}\ge 1, \cdots, \abs{\lambda_n} \ge 1$, and
  $\abs{\lambda_i} >1$ for some $i$, it holds that
  $P(\lambda_1,\cdots,\lambda_n)\neq 0$.
\end{definition}

Then the seminal Lee-Yang theorem~\cite{leeyan52b} can be stated as follows:
\begin{theorem}
	Let $G$ be a connected undirected graph, and suppose $0<\beta<1$.  Then the Ising partition function
	$Z_G^{\beta}(\lambda_1, \cdots, \lambda_n)$ has the Lee-Yang property.
	\label{thm:leeyang-graph}
\end{theorem}

The following extension of the Lee-Yang theorem to general symmetric 
two-spin systems on hypergraphs
is due to Suzuki and Fisher~\cite{suzuki1971zeros}.  Again the theorem is stated in the
multivariate setting, where in the two-spin partition function in~eq.\nobreakspace \textup {(\ref {eqn:hyperZ})} each
vertex~$i$ has a distinct activity~$\lambda_i$.
\begin{theorem}
Consider any symmetric hypergraph two-spin system, with a connected hypergraph $G$ and edge 
activities $\set{\varphi_e}$.  Then the partition function
	$Z_G^{\bm \varphi}(\lambda_1, \cdots, \lambda_n)$ has the Lee-Yang property if $\abs{\varphi_e(+,\cdots, +)} \ge \frac{1}{4}\sum_{\bm \sigma\in \set{+,-}^V} \abs{\varphi_e(\bm \sigma)}$ for every hyperedge $e$.
	\label{thm:suzuki-fisher}
\end{theorem}

\MakeUppercase Theorem\nobreakspace \ref {thm:suzuki-fisher} is not tight for the important special case of the Ising model on
hypergraphs.  Our goal in this section is to prove a tight analog of the original Lee-Yang
theorem for this case.
Specifically, we will prove the following:

\begin{theorem}
  Let $G=(V,E)$ be a connected hypergraph, and
  $\vec{\beta} = (\beta_e)_{e \in E}$ be a vector of real valued
  hyperedge activities so that the activity of edge $e \in E$ is
  $\beta_e$.  Then $Z_G^{\vec{\beta}}$ has the Lee-Yang property if
  the following condition holds for every hyperedge $e$, where $k\ge 2$ is
  the size of $e$:
	\begin{itemize}
    \item if $k =2$, then $-1 < \beta_e < 1$;
    \item if $k \geq 3$, then
      $-\frac{1}{2^{k-1} -1} <\beta_e < \frac{1}{2^{k-1}\cos^{k-1}\left(
          \frac{\pi}{k-1} \right) +1}$.
	\end{itemize}
    Further, if the above condition is not satisfied for a
    given real edge activity $\beta$ and integer $k \geq 2$, then
    there exists a $k$-uniform hypergraph $H$ with edge activity~$\beta$ such that $Z_H^\beta$
    does not have the Lee-Yang property.
	\label{thm:hyper-leeyang}
\end{theorem}
Note that the case $k=2$ is just the original Lee-Yang theorem (Theorem~\ref{thm:leeyang-graph}).

The following corollary for the univariate polynomial $Z_G^{\vec{\beta}}(\lambda)$
follows immediately via~eq.\nobreakspace \textup {(\ref {eqn:multivariate-symmetry})}
and the fact that, by Hurwitz's theorem, the zeros of $Z_G^{\vec{\beta}}(\lambda)$ are continuous functions of $\vec{\beta}$ and thus remain on the unit circle after taking the limit
in the range of each~$\beta_e$.
\begin{corollary}
  Let $G=(V,E)$ be a connected hypergraph, and $\vec{\beta} = (\beta_e)_{e \in E}$ be the vector of real valued
  hyperedge activities so that the activity of edge $e \in E$ is
  $\beta_e$.  Then, all complex zeros of the univariate partition
  function $Z_G^{\vec{\beta}}(\lambda)$ lie on the unit circle if
  the following condition holds for every hyperedge $e$, where $k\ge 2$ is the size of $e$:
	\begin{itemize}
    \item if $k =2$, then $-1 \leq \beta_e \leq 1$;
    \item if $k \geq 3$, then
      $-\frac{1}{2^{k-1} -1} \le \beta_e \le \frac{1}{2^{k-1}\cos^{k-1}\left(
          \frac{\pi}{k-1} \right) +1}$.
	\end{itemize}
	\label{cor:hyper-leeyang}
\end{corollary}
The corollary establishes the first part of
Theorem~\ref{thm:leeyangintro} in the introduction, and hence also
Theorem~\ref{thm:main2} as explained at the end of the previous
section.  The second part of Theorem~\ref{thm:leeyangintro}, which
asserts that the range of edge activities under which the theorem
holds is optimal, is proven in Section~\ref{sec:tight-ly}.
(Note that the optimality for the univariate case claimed in
Theorem~\ref{thm:leeyangintro} does not directly follow from the
optimality for the multivariate case guaranteed by
Theorem~\ref{thm:hyper-leeyang} above.)

\begin{remark}
As a comparison, the result of Suzuki and Fisher, 
which we restated in~\MakeUppercase Theorem\nobreakspace \ref {thm:suzuki-fisher}, implies that a sufficient condition
for the Lee-Yang property of $Z_G^{\vec{\beta}}(\lambda)$ is
	\[
		-\frac{1}{2^{k-1} -1} \le \beta_e \le \frac{1}{2^{k-1} -1}.
	\]
	Note that while the lower bound on $\beta_e$ is the same as ours,
    our (tight) upper bound is always better, and significantly so for
    the more interesting case of small $k$.  For example, for $k=3$ our
    result gives the optimal range $-\frac{1}{3} \le \beta_e \le 1$,
    while the Suzuki-Fisher theorem gives
    $-\frac{1}{3} \le \beta_e \le \frac{1}{3}$.  Similarly, for
    $k = 4$ the respective ranges are $[-1/7, 1/2]$ (for ours) and
    $[-1/7, 1/7]$ (for Suzuki-Fisher).
\end{remark}
We turn now to the proof of Theorem~\ref{thm:hyper-leeyang}.  The main
technical step in our proof is to derive conditions under which the
Ising partition function of a hypergraph consisting of a {\it single
  hyperedge\/} has the Lee-Yang property. This ``base case'' turns out
to be more difficult than in the case of the original Lee-Yang theorem
for graphs.  However, as in the graph case, it will turn out that the
base case still determines the range of $\beta$ in which the
theorem can be claimed to be valid; we show this latter claim, which
implies the second part of Theorem~\ref{thm:hyper-leeyang} in
Section~\ref{sec:tight-ly}.

We begin with the following two lemmas which, taken together, give a
partial characterization of the Lee-Yang property.  While similar in
spirit to the results of Ruelle~\cite{ruelle_characterization_2010},
these lemmas do not follow directly from those results since, as noted
above, the version of the Lee-Yang property used here imposes a
stricter condition on the polynomial than does the definition used
in~\cite{ruelle_characterization_2010}. %

\begin{lemma}
  Given a multilinear polynomial $P(z_1, z_2, \dots, z_n)$ with real
  coefficients define, for each $1 \leq j \leq n$, multilinear
  polynomials $A_j$ and $B_j$ in the variables $z_1, \dots, z_{j-1},
  z_{j+1}, \dots, z_n$ such that
  \[
    P = A_j z_j + B_j.
  \]
If $P$ has the Lee-Yang property then, for every $j$ such that the
  variable $z_j$ has positive degree in $P$, it holds that
  $A_j(z_1, \dots, z_{j-1}, z_{j+1}, \dots, z_n) \neq 0$ when
  $\abs{z_i} \ge 1$ for all $i \neq j$.  In particular, $A_j$ itself is
  LY.
  \label{lem:asano}
\end{lemma}
\begin{proof}
  Without loss of generality, we assume that $j = 1$.  Note that since
  $z_1$ has positive degree in $P$, $A_1$ is a non-zero polynomial.
  Suppose that, in contradiction to the claim of the lemma, there
  exist complex numbers $\lambda_2, \dots, \lambda_n$ satisfying
  $\abs{\lambda_i} \geq 1$ such that
  $A_1(\lambda_2, \dots, \lambda_n) = 0$.  Since $P$ is LY, it follows
  that $B_1(\lambda_2, \dots, \lambda_n) \neq 0$ (for otherwise, we
  get a contradiction to the Lee-Yang property by choosing $z_1$ to be
  an arbitrary value outside the closed unit disk).

  By continuity, this implies that $\abs{B_1}$ is positive in any
  small enough neighborhood of $(\lambda_2, \dots, \lambda_n)$ in
  $\C^{n-1}$.  In particular, let $S_\varepsilon$ be the open set
  \[
    S_\varepsilon \defeq \inbr{ (y_2, \dots, y_n) \st \abs{y_i - \lambda_i} <
      \varepsilon \text{ and } \abs{y_i} > 1 \text{ for $2 \leq i \leq n$ }}.
  \]
  Then there exist positive $\delta_0$ and $\varepsilon_0$ such that
  $\abs{B_1}$ is at least $\delta_0$ in the open set $S_\varepsilon$ when
  $\varepsilon < \varepsilon_0$.

  Now, since $A_1$ is a non-zero multilinear polynomial, it cannot
  vanish identically on any open set.  In particular, it cannot vanish
  identically in $S_\varepsilon$ for any $\varepsilon > 0$.  On the other
  hand, since $A_1$ vanishes at $(\lambda_2, \dots, \lambda_n)$ it
  follows from continuity that for $\varepsilon < \varepsilon_0$ small
  enough, $\abs{A_1} \leq \delta_0/2$ in $S_\varepsilon$.  Since $A_1$
  does not vanish identically on $S_\varepsilon$, there must exist a
  point $(y_2, \dots, y_n)$ in $S_\varepsilon$ such that
  $0 < \abs{A_1(y_2, \dots, y_n)} < \delta_0/2$.  Since
  $\abs{B_1(y_2, \dots, y_n)} \geq \delta_0$ by the choice of
  $\varepsilon_0$, it follows that if we define
  $y_1 = -B_1(y_2, \dots, y_n)/A_1(y_2, \dots, y_n)$ then
  $2 < \abs{y_1} < \infty$.  However, we then have
  $P(y_1, y_2, \dots, y_n) = 0$ even though $\abs{y_1} > 1$ and
  $\abs{y_i} \geq 1$ for all $i$.  This contradicts the Lee-Yang
  property of $P$.
\end{proof}

By iterating the above lemma, we get the following corollary.
\begin{corollary}
  \label{cor:asano-iter}
  Let $P(z_1, z_2, \dots, z_n)$ be a multilinear polynomial with
  non-zero real coefficients, i.e.,
  \begin{displaymath}
    P(z_1, \dots, z_n) = \sum_{S \subseteq [n]} p_S\prod_{i \in S}z_i,
  \end{displaymath}
  where $p_S \in \R$ are non-zero for all $S \subseteq [n]$, and assume that $P$ is LY.  Then, for every subset
  $S$ of $[n]$, the polynomial $A_S$ defined by the equation
  \begin{displaymath}
    P(z_1, \dots, z_n) = A_S\inp{(z_i)_{i \not\in S}}
    \prod_{i \in S}z_i + \sum_{T:\,S \not \subseteq T} p_T\prod_{i \in T}z_i
  \end{displaymath}
  has the property that $A_S\inp{(z_i)_{i \not\in S}} \neq 0$ when
  $\abs{z_i} \geq 1$ for all $i \not\in S$.  In particular, $A_S$ is LY.
\end{corollary}

We next show that \MakeUppercase Lemma\nobreakspace \ref {lem:asano} has a partial converse for
symmetric multilinear functions.
\begin{lemma}
  Let $P(z_1, z_2, \dots, z_n)$ be a symmetric multilinear polynomial
  with non-zero real coefficients, i.e.,
  \begin{displaymath}
    P(z_1, \dots, z_n) = \sum_{S \subseteq [n]} p_S\prod_{i \in S}z_i,
  \end{displaymath}
  where $p_S \neq 0$ for all $S \subseteq [n]$ and
  $p_S = p_{\overline{S}}$.
  Assume further that the polynomials $A_j$
  as defined in \MakeUppercase Lemma\nobreakspace \ref {lem:asano} all have the property that they are
  non-zero when all their arguments $z_i$ satisfy $\abs{z_i} \ge 1$.
  Then $P$ is LY.
  \label{lem:asano-leeyang-eqv}
\end{lemma}

\begin{proof}
  We first show that, under our assumptions, if all but one of the
  $z_j$ lie on the unit circle, then $P$ can only vanish if the
  remaining $z_j$ is also on the unit circle.  Without loss of
  generality we set $j=1$, that is, we will show that if
  $\abs{z_i} = 1$ for $i\ge 2$, then any root $z_1=\zeta_1$ of the
  equation $A_1z_1 + B_1 = 0$ satisfies $\abs{\zeta_1} = 1$.  (Here
  $A_1$ and $B_1$ are in the notation of \MakeUppercase Lemma\nobreakspace \ref {lem:asano}.)

  Since by assumption
  $A_1 = \sum_{S \subseteq [2,n]}p_{S \cup \inbr{1}}\prod_{i \in
    S}z_i$ does not vanish with this setting of the $z_i$, we have
  \begin{align}
    \abs{\zeta_1} = \abs{
    \frac{B_1}{A_1}
    }
    &= \abs{
      \frac{
      \sum_{S \subseteq [2,n]}p_{S}\prod_{i \in S}z_i
      }{
      \sum_{S \subseteq [2,n]}p_{S \cup \{1\}}\prod_{i \in S}z_i
      }}
      = \abs{
      \inp{\prod_{i \in [2, n]} z_i}
      \frac{
      \sum_{S \subseteq [2,n]} p_{S}
      \prod_{\substack{i \not\in S\\i \neq 1}}(1/z_i)
    }{
    \sum_{S \subseteq [2,n]}p_{S \cup \{1\}}\prod_{i \in S}z_i
    }}\nonumber \\
    &\overset{(\star)}{=}
      \abs{
      \frac{
      \sum_{S \subseteq [2,n]} p_{S}
      \prod_{\substack{i \not\in S\\i \neq 1}}
    \overline{z_i}
    }{
    \sum_{S \subseteq [2,n]}p_{S \cup \{1\}}\prod_{i \in S}z_i
    }}
    \overset{(\dagger)}{=}
    \abs{
    \frac{
    \sum_{S \subseteq [2,n]} p_{S \cup \{1\}}
    \prod_{i \in S}\overline{z_i}
    }{
    \sum_{S \subseteq [2,n]}p_{S \cup \{1\}}\prod_{i \in S}z_i
    }} = 1.\label{eq:1}
  \end{align}
  Here $(\star)$ uses the fact that $\abs{z_i} = 1$ for $i \geq 2$ and
  $(\dagger)$ uses the symmetry of $P$.
We have thus shown that if $(z_1, z_2, \dots, z_n)$ is a zero of $P$
  such that $\abs{z_i} \geq 1$ for all $i$ then it is impossible for
  only one $z_i$ to lie outside the closed unit disk.

  We now show
  that if there are $k \geq 2$ values of $i$ for which $z_i$ lies
  outside the closed unit disk, then we can find another zero
  $(\zeta_1, \zeta_2, \zeta_3, \dots, \zeta_n)$ of $P$ such that
  $\abs{\zeta_i} \geq 1$ for all $i$, and exactly $k - 1$ of the
  $\zeta_i$ lie outside the closed unit disk.  We can then iterate
  this process to reduce $k$ to $1$, in which case the observation
  in the previous paragraph leads to a contradiction.

  By re-numbering the indices if needed, we can assume that
  $\abs{z_1}, \abs{z_2} > 1$ and $\abs{z_i} \geq 1$ for $i \geq 3$.
  We can then write
  \begin{displaymath}
    P(z_1, \dots, z_n) = \alpha_{12} z_1z_2 +
    \alpha_{1}z_1 + \alpha_{2}z_2 + \alpha_{\emptyset},
  \end{displaymath}
  where $\alpha_{12}, \alpha_{1}, \alpha_{2}$ and $\alpha_{\emptyset}$
  are non-zero polynomials in $z_3, \dots, z_n$.  Further, the
  hypotheses of the lemma imply that $A_1 = \alpha_{12}z_2 + \alpha_1$
  and $A_2 = \alpha_{12}z_1 + \alpha_{2}$ both have the Lee-Yang
  property.  Thus, by \MakeUppercase Lemma\nobreakspace \ref {lem:asano},
  $\alpha_{12}(z_3, \dots, z_n) \neq 0$ when $\abs{z_i} \geq 1$ for
  $i \geq 3$.  Now, again by hypothesis, $A_2 \neq 0$ when $\abs{z_1}$ and
  $\abs{z_3}, \dots, \abs{z_n}$ are at least $1$, while
  $z_1 = -\frac{\alpha_2(z_3, \dots, z_n)}{\alpha_{12}(z_3, \dots,
    z_n)}$ gives $A_2=0$.  Thus, we must have that
  \begin{equation}
    \frac{
      \abs{\alpha_2(z_3, \dots, z_n)}
    }{
      \abs{\alpha_{12}(z_3, \dots, z_n)}
    } < 1
    \text{ when $\abs{z_i} \geq 1$ for $i \geq 3$}.\label{eq:2}
  \end{equation}
  We now set $\zeta_i = z_i$ for $i \geq 3$, and consider $z_1$ as a
  function of $z_2$.  The equality $P(z_1, z_2, \zeta_3, \dots, \zeta_n) = 0$
  is then equivalent to
  \begin{equation}
    \label{eq:3}
    z_1 = -\frac{\alpha_2 z_2 + \alpha_\emptyset}{\alpha_{12}z_2 + \alpha_1},
  \end{equation}
  where the hypotheses of the lemma imply that the denominator
  (which is equal to $A_1(z_2, \zeta_3, \dots, \zeta_n)$) is non-zero when
  $\abs{z_2} \geq 1$.  We thus see that
  \begin{equation}
    \lim_{z_2 \rightarrow \infty}\abs{z_1}
    = \frac{\abs{\alpha_2}}{\abs{\alpha_{12}}}
    < 1. \label{eq:4}
  \end{equation}
  Initially, both $z_1$ and $z_2$ lie outside the closed unit disk.
  Thus, by eq.\nobreakspace \textup {(\ref {eq:4})} and continuity, we can take $z_2$ large enough
  in absolute value such that $z_1$ as defined in eq.\nobreakspace \textup {(\ref {eq:3})} lies on
  the unit circle.  We now choose $\zeta_1$ and $\zeta_2$ to be these
  values of $z_1$ and $z_2$, respectively, so that we have
  $P(\zeta_1, \dots, \zeta_n) = 0$ and the number of the $\zeta_i$
  lying on the unit circle is exactly one less than the number of the
  $z_i$ lying on the unit circle, as required.
\end{proof}

Along with the above general facts about LY polynomials, we also need
the following technical lemma.

\begin{lemma}
  Let $m$ be any integer, and $k$ a positive integer such that
  $2\abs{m} \leq k$.  Consider the maximization problem
\begin{align*}
	\max\; &\prod_{i=1}^{k} \cos \theta_i \\
	\textrm{subject to }\; &-\frac{\pi}{2} \le \theta_i \le \frac{\pi}{2}, \\
	&\sum_{i=1}^{k} \theta_i = m\pi.
\end{align*}
The maximum is $\cos^{k} \left( \frac{m\pi}{k} \right)$, and is
attained when $\theta_i=\frac{m\pi}{k}$ for all~$i$.
	\label{lem:logcos}
\end{lemma}
\begin{proof}
  We may assume without loss of generality that
  $\theta_i \in (-\pi/2, \pi/2)$ at any maximum (for otherwise the
  objective value is $0$).  Now, consider the function
  $f(x) = \log \cos x$ defined on the interval $(-\pi/2,
  \pi/2)$. Since $f'(x) = - \tan x$ is a decreasing function, $f(x)$
  is concave for $x \in (-\frac{\pi}{2}, \frac{\pi }{ 2})$.  Thus by
  Jensen's inequality,
	\[
      \log \prod_{i=1}^{k} \cos \theta_i = \sum_{i=1}^{k}
      f(\theta_i) \le k f\left(\frac{\sum_{i=1}^{k}
          \theta_i}{k} \right) \le k \log\cos \left(
        \frac{m\pi}{k} \right),
	\]
	and equality holds when $\theta_i = \frac{m\pi}{k}$ for all~$i$.
    Note that these $\theta_i$ are in $(-\pi/2, \pi/2)$ since
    $2\abs{m} \leq k$.
\end{proof}

We are now ready to tackle the case of a single hyperedge.
\begin{lemma}
  Fix an integer $k\geq 2$ and a hyperedge activity $\beta \in \R$.
  Let
  $G = (V = \inbr{v_1, v_2, \dots, v_k}, E = \{\inbr{v_1, v_2,
      \dots, v_k}\})$ be a hypergraph consisting of a single hyperedge
  of size $k$ and activity $\beta$.  If $k = 2$ and
  $\beta \in (-1, 1)$, or $k \geq 3$ and $\beta$ satisfies
  \[
    -\frac{1}{2^{k-1} -1} <\beta < \frac{1}{2^{k-1}\cos^{k-1}\left( \frac{\pi}{k-1} \right) +1},
  \]
  then the partition function $Z_G^\beta$ has the Lee-Yang property.
  \label{lem:leeyang-base}
\end{lemma}

\begin{remark}
  Note that the condition on $\beta$ imposed above is monotone in $k$:
  i.e., if $\beta$ is such that the partition function of a hyperedge of
  size $k \geq 2$ is LY, then for the same $\beta$ the partition
  function of a hyperedge of size $k' < k$ is also LY.
\end{remark}

\begin{proof}
  For $k = 2$, the lemma is a special case of the Lee-Yang
  theorem~\cite{leeyan52b} (although it also follows by specializing
  the argument below).  We therefore assume $k \geq 3$.

  Since the Ising partition function is symmetric and all terms in the
  polynomial appear with positive coefficients,
  \MakeUppercase Lemma\nobreakspace \ref {lem:asano-leeyang-eqv} applies and it suffices to verify that
  the polynomials $A_j$ do not vanish when $\abs{z_i}\geq 1$ for
  $i \neq j$.  Without loss of generality we fix $j = 1$.  We then
  have
  \[
    A_1 = \beta \prod_{i = 2}^k (1+ z_i) + (1
    - \beta) \prod_{i = 2}^k z_i.
  \]
  Thus $A_1 = 0$ is equivalent to
  \begin{equation}
    \frac{1}{\beta} = 1 - \prod_{i = 2}^k \inp{ 1
      + \frac{1}{z_i}}.\label{eq:5}
    \end{equation}
  To establish the lemma, we therefore only need to show that for the
  claimed values of $\beta$, eq.\nobreakspace \textup {(\ref {eq:5})} has no solutions when
  $\abs{z_i} \geq 1$ for all $i \geq 2$. We now proceed to establish
  this by analyzing the product on the right hand side of eq.\nobreakspace \textup {(\ref {eq:5})}.

  The map $z \mapsto 1 + 1/z$ is a bijection from the complement of
  the open unit disk to the closed disk $D$ of radius $1$ centered at
  $1$.  Any $y \in D$ can be written as $y=r\exp\inp{\iota \theta}$ for
  $\theta \in [-\pi/2, \pi/2]$ and $0 \le r \leq 2\cos \theta$. Consider now
  the set
  \(\R \cap \bigl\{\prod_{i=2}^{k}y_i \st y_i \in D\text{ for $2\le i\le k$}\bigr\}\).
  We show that, for $k\ge 3$, this set is exactly
  the interval $[-\tau_0, \tau_1]$ where
  $\tau_0 = 2^{k-1}\cos^{k-1}(\pi/(k-1))$ and $\tau_1 = 2^{k-1}$.  The
  claim of the lemma then follows since for the given values of
  $\beta$, $1 - 1/\beta$ lies outside $[-\tau_0, \tau_1]$ and hence
  eq.\nobreakspace \textup {(\ref {eq:5})} cannot hold.

  Recalling that each $y \in D$ can be written in the form
  $r \exp(\iota \theta)$ where $\theta \in [-\pi/2, \pi/2]$ and
  $0 \leq r \leq 2 \cos \theta$, we find that the values $\tau_0$ and
  $\tau_1$ are defined by the following optimization problems (both of which are
  feasible since $k \geq 3$):

  \begin{displaymath}
    \begin{aligned}
      \tau_0 = 2^{k-1}\max &\prod_{i=2}^{k} \cos \theta_i \\
      \textrm{subject to }\; &-\frac{\pi}{2} \le \theta_i \le \frac{\pi}{2}, \\
      &\sum_{i=2}^{k} \theta_i = (2n+1)\pi\\
      &\text{for some $n \in \mathbb{Z}$}\\[-4pt]
      &\text{s.t. $\abs{2n + 1} \leq (k-1)/2$}.
    \end{aligned}
    \begin{aligned}
      \tau_1 = 2^{k-1} \max &\prod_{i=2}^{k} \cos \theta_i \\
      \textrm{subject to }\; &-\frac{\pi}{2} \le \theta_i \le \frac{\pi}{2}, \\
      &\sum_{i=2}^{k} \theta_i = 2n\pi\\
      &\text{for some $n \in \mathbb{Z}$}\\[-4pt]
      &\text{s.t. $\abs{n} \leq (k-1)/4$.}\\[12pt]
    \end{aligned}
  \end{displaymath}

  Using \MakeUppercase Lemma\nobreakspace \ref {lem:logcos}, we then see that
  $\tau_0 = 2^{k-1}\cos^{k-1}(\pi/(k-1))$ and $\tau_1 = 2^{k-1}$, as
  required.
\end{proof}

We now proceed to an inductive proof of Theorem~\ref{thm:hyper-leeyang}, using
Lemma~\ref{lem:leeyang-base} as the base case.
\begin{proof}[Proof of \MakeUppercase Theorem\nobreakspace \ref {thm:hyper-leeyang}]
  The case $k = 2$ is a special case of the Lee-Yang
  theorem~\cite{leeyan52b} (though, as with the proof of Lemma~\ref{lem:leeyang-base},
  the argument below can
  again be specialized to directly establish this).  We assume
  therefore that $k \geq 3$.

  The proof uses the inductive method of
  Asano~\cite{asano_lee-yang_1970}.  When the hypergraph consists of a
  single hyperedge of size $k' \leq k$, it follows from
  \MakeUppercase Lemma\nobreakspace \ref {lem:leeyang-base} and the remark immediately after it that the
  partition function is LY for the claimed
  values of the edge activity~$\beta$.
  For the induction, we use the fact that the Lee-Yang property of the
  partition function is preserved under the following two operations:
  \begin{enumerate}
  \item \textbf{Adding a hyperedge}: In this operation, a new
    hyperedge~$e$ of size $k' \leq k$ and activity $\beta_e$ as claimed in
    the statement of the theorem, is added to a connected hypergraph
    in such a way that exactly one of its $k'$ vertices already exists
    in the starting hypergraph, while the other $k' - 1$ vertices are
    new.  Note that this operation keeps the hypergraph connected.  We
    assume that the partition functions of both the original
    hypergraph as well as the newly added edge separately have the
    Lee-Yang property: this follows from the induction hypothesis (for
    the hypergraph) and \MakeUppercase Lemma\nobreakspace \ref {lem:leeyang-base} (for the new
    hyperedge).
  \item \textbf{Asano contraction}: In this operation, two vertices
    $u', u''$ in a connected hypergraph that are not both included in
    any one hyperedge are merged so that the new merged vertex $u$ is
    incident on all the hyperedges incident on $u'$ or $u''$ in
    the original graph.  Note that this operation keeps the hypergraph
    connected and does not change the size of any of the
    hyperedges.
  \end{enumerate}
  Any connected non-empty hypergraph $G$ can be constructed by
  starting with any arbitrary hyperedge present in $G$ and performing
  a finite sequence of the above two operations: to add a new
  hyperedge $e$ with activity $\beta_e$, one first
  uses operation~1 to add a hyperedge which has the same activity
  $\beta_e$ and has new copies of all but one of the incident vertices
  of $e$, and then uses operation~2 to merge these new copies with
  their counterparts, if any, in the previous hypergraph. Note that in
  this process, a hyperedge $e$ can be added only when at least one
  of its vertices is already included in the current hypergraph.
  However, since $G$ is assumed to be connected, its hyperedges can be
  ordered so that all of them are added by the above process.  Thus,
  assuming that the above two operations preserve the Lee-Yang
  property, it follows by induction 
  that the partition functions of all connected hypergraphs of
  hyperedge size at most $k$, and edge activities $\beta_e$ as claimed
  in the theorem, have the Lee-Yang property.

  Given \MakeUppercase Corollary\nobreakspace \ref {cor:asano-iter}, it can be proved, by adapting an
  argument first developed by Asano~\cite{asano_lee-yang_1970}, that
  these two operations preserve the Lee-Yang property.  Asano's method
  has by now become standard (see, e.g., \cite[Propositions 1,
  2]{suzuki1971zeros}), but we include the details here for
  completeness.

  Consider first operation~1.  Let $G$ be the original
  hypergraph and $H$ the new hyperedge (with $k' \leq k$ vertices)
  being added, and assume, by renumbering vertices if required, that
  the single shared vertex is $v_1$ in $G$ and $u_1$ in $H$,
  respectively.  Let
  $P(z_1, z_2, \dots, z_n) = A(z_2, \dots, z_n)z_1 + B(z_2, \dots,
  z_n)$ and
  $Q(y_1, y_2, \dots, y_{k'}) = C(y_2, \dots, y_{k'})y_1 + D(y_2,
  \dots, y_{k'})$ be the Ising partition functions of $G$ and $H$,
  respectively, where $z_1$ and $y_1$ are the variables corresponding
  to $v_1$ and $u_1$, respectively.  Both $P$ and $Q$ are LY by the
  hypothesis of the operation.  The partition function $R$ of the new
  graph can be written as
  \begin{displaymath}
    R(z, z_2, \dots, z_n, y_2, \dots, y_{k'}) = A(z_2, \dots,
    z_n)C(y_2, \dots, y_{k'}) z + B(z_2, \dots, z_n)D(y_2, \dots, y_{k'}),
  \end{displaymath}
  where $z$ is a new variable corresponding to the new vertex created
  by the merger of $u_1$ and $v_1$.  Let
  $\lambda_2, \dots, \lambda_n, \mu_2, \dots, \mu_{k'}$ be complex
  numbers lying outside the open unit disk.  In order to prove that
  $R$ is LY, we need to show that
(i)~$R(z, \lambda_2, \dots, \lambda_n, \mu_2, \dots, \mu_{k'}) = 0$
  implies that $\abs{z} \leq 1$; and (ii)~when at least one of these
  complex numbers lies strictly outside the \emph{closed} unit disk
  then $R(z, \lambda_2, \dots, \lambda_n, \mu_2, \dots, \mu_{k'}) = 0$
  implies that $\abs{z} < 1$.
  Now, since $P$ and $Q$ are assumed to be LY, \MakeUppercase Lemma\nobreakspace \ref {lem:asano}
  implies that $A = A(\lambda_2, \dots, \lambda_n)$ and
  $C = C(\mu_2, \dots, \mu_{k'})$ are both non-zero.  Thus, $R = 0$
  implies that
  \begin{equation}
    \abs{z} = \abs{B/A}\cdot\abs{D/C},\label{eq:7}
  \end{equation}
  where $B = B(\lambda_2, \dots \lambda_n)$ and
  $D = D(\mu_2, \dots, \mu_{k'})$.  Since all the $\lambda_i$ and
  $\mu_i$ lie outside the open unit disk and $P$ and $Q$ are LY,
  $\abs{B/A}, \abs{D/C} \leq 1$, so that from eq.\nobreakspace \textup {(\ref {eq:7})},
  $\abs{z} \leq 1$.  This establishes condition~(i).
  Further, when at least one of the $\lambda_i$
  lies strictly outside the closed unit disk, then again, since $P$ is
  LY, $\abs{B/A} < 1$. Similarly,$\abs{D/C} < 1$ when one of the
  $\mu_i$ lies outside the closed unit disk.  Thus, when at least one
  of the $\lambda_i$ and the $\mu_i$ lies outside the closed unit
  disk, it follows from eq.\nobreakspace \textup {(\ref {eq:7})} that $\abs{z} < 1$, thus establishing
  condition~(ii) and concluding the argument that $R$ is LY.

  We now consider operation~2.  By renumbering vertices if
  necessary, let $v_1$ and $v_2$ be the vertices to be merged.  The
  partition function $P$ of the original graph (where $v_1$ and $v_2$
  are not merged) can be written as
  \begin{displaymath}
    P(z_1, z_2, z_3, \dots, z_n) = A(z_3, \dots, z_n)z_1z_2 +
    B(z_3, \dots, z_n)z_1 + C(z_3, \dots, z_n)z_2 + D,
  \end{displaymath}
  and is LY by the hypothesis of the operation.  The partition
  function $R$ after the merger is then given by
  \begin{displaymath}
    R(z, z_3, \dots, z_n) = A(z_3, \dots, z_n)z + D,
  \end{displaymath}
  where $z$ is a new variable corresponding to the vertex created
  by the merger of $v_1$ and $v_2$. Now, let
  $\lambda_3, \dots, \lambda_n$ be complex numbers lying outside the
  open unit disk.  \MakeUppercase Corollary\nobreakspace \ref {cor:asano-iter} implies that
  $A = A(\lambda_3, \dots, \lambda_n) \neq 0$.  Thus, $R(z, \lambda_3,
  \dots, \lambda_n) = 0$ implies that
  \begin{equation}
    \label{eq:10}
    \abs{z}  = \abs{D(\lambda_3, \dots,
      \lambda_n)/A(\lambda_3, \dots, \lambda_n)}  = \abs{D/A}.
  \end{equation}
  Now, since $P$ is LY, both zeros of the quadratic equation
  $P(x, x, \lambda_3, \dots, \lambda_n) = 0$ satisfy $\abs{x} \leq 1$,
  and indeed, $\abs{x} < 1$ when at least one of the $\lambda_i$ lies
  strictly outside the closed unit disk.  Thus, the product $D/A$ of
  its zeros also satisfies $\abs{D/A} \leq 1$, and further satisfies
  the stronger inequality $\abs{D/A} < 1$ in case at least one of the
  $\lambda_i$ lies strictly outside the closed unit disk.
  \MakeUppercase eq.\nobreakspace \textup {(\ref {eq:10})} then implies that $\abs{z} \leq 1$ in the first case
  and $\abs{z} < 1$ in the second case, which establishes that $R$ is
  LY.

  This concludes the proof of the first part of
  Theorem~\ref{thm:hyper-leeyang}. We now prove the optimality of the
  conditions imposed on the edge parameters. In the case $k = 2$, this
  follows by considering the partition function
  $z_1z_2 + \beta z_1 + \beta z_2 + 1$ of a single edge. When
  $\beta>1$ (respectively, when $\beta<-1$),  $z_1 = z_2 = -\beta - \sqrt{\beta^2 - 1}$ 
  (respectively, $z_1 = z_2 = -\beta + \sqrt{\beta^2 - 1}$) is a zero
    of the partition function satisfying $\abs{z_1}, \abs{z_2} > 1$ and
  hence contradicting the Lee-Yang property.  Similarly
  $z_1 = -1, z_2 = 2$ when $\beta = 1$, and $z_1 = 1, z_2 = 2$ when
  $\beta = -1$, are zeros which contradict the Lee-Yang property.

  We now consider the case $k \geq 3$.  In this case, we take our
  example to be the single hyperedge of size $k$ and consider its
  partition function
  \begin{equation}
    P(z_1, z_2, \dots, z_k) \defeq \beta\prod_{i=1}^k(1+z_i) +
    (1-\beta)\inp{1 + \prod_{i=1}^kz_i}.\label{eq:8}
  \end{equation}
  Our strategy is to show that when
  \begin{equation}
    \beta \not\in \inp{-\frac{1}{2^{-k-1} - 1},
      \frac{1}{2^{k-1}\cos^{k-1}\bigl(\frac{\pi}{k-1}\bigr) + 1}},\label{eq:6}
  \end{equation}
  the polynomial $A_1(z_2, z_2, \dots, z_k)$, which is the coefficient
  of $z_1$ in $P$ as defined in Lemma~\ref{lem:asano}, vanishes at a
  point with $\abs{z_i} \geq 1$ for $i \geq 2$.  It then follows from
  Lemma~\ref{lem:asano} that $P$ cannot have the Lee-Yang property.

  To carry out the strategy, we reuse some of the notation and
  calculations from the proof of Lemma~\ref{lem:leeyang-base}.  Let
  $D$ be the closed disk of radius $1$ centered at $1$, as defined in
  the proof of that lemma.  Eq.~\eqref{eq:5}, taken together with the
  discussion following it, implies that finding a zero of
  $A_1(z_2, \dots, z_k)$ with $\abs{z_i} \geq 1$, $2 \leq i \leq k$,
  is equivalent to finding $y_i \in D$, $y_i \neq 1$ such that
  $1 - \frac{1}{\beta}= \prod_{i=2}^{k}y_i$.  We can in fact choose
  all the $y_i$ to be equal, so that using the same representation of $D$
  as in the proof of Lemma~\ref{lem:leeyang-base}, our task reduces to finding an angle
  $\theta \in [-\pi/2, \pi/2]$ and $0 \leq r \leq 2\cos\theta$ such that
  $y_i = r e^{\iota\theta}$, $2 \leq i \leq k$, and
  \begin{equation}
    \label{eq:9}
    1 - \frac{1}{\beta} =
    \bigl(r e^{\iota\theta}\bigr)^{k-1}.
  \end{equation}
  Let $\gamma \defeq 1 - \frac{1}{\beta}$.  We now partition the
  condition on $\beta$ in \eqref{eq:6} into three different cases.
  Suppose first that $\beta \leq -\frac{1}{2^{k-1} - 1}$.  This is
  equivalent to $1 < \gamma \leq 2^{k-1}$. In this case $\theta = 0$,
  $r = \gamma^{\frac{1}{k-1}} \in (1, 2]$
  gives a desired solution to \eqref{eq:9} (note that we have
  $y_i \in (1, 2]$ in this case).  The same solution for $\theta$ and
  $r$ also works when $\beta > 1$ (in which case $0 < \gamma < 1$ and
  $y_i \in (0, 1)$) .  The remaining case is
  $1 \geq \beta \geq \frac{1}{2^{k-1}\cos^{k-1}\inp{\frac{\pi}{k-1}} +
    1}$, which in turn is equivalent to
  $-2^{k-1}\cos^{k-1}(\frac{\pi}{k-1}) \leq \gamma \leq 0$, and
  $\theta = \frac{\pi}{k-1}$, $r =  \abs{\gamma}^{\frac{\pi}{k-1}} \le 2\cos\theta$
  \ignore{
  \[
    \theta = \frac{\pi}{k-1},%
    \;\; r = \inp{%
      \frac{%
        \abs{\gamma}%
      }{%
        2^{k-1}\cos^{k-1}%
        \inp{%
          \frac{\pi}{k-1}%
        }%
      }%
    }^{\frac{1}{k-1}}%
    \leq 1
  \]
  }
  gives a solution in this case. \qedhere
\end{proof}

\subsection{Optimality of the univariate hypergraph Lee-Yang theorem}
\label{sec:tight-ly}

We now prove the second part of the univariate hypergraph Lee-Yang
theorem, Theorem~\ref{thm:leeyangintro}, i.e., that the range of edge
activities under which the first part of that theorem holds is
optimal.  The tight example for the case $k = 2$ is again a single
edge, and as observed above, the roots of the univariate partition
function of the edge when $\abs{\beta} > 1$ are
$-\beta \pm \sqrt{\beta^2 - 1}$, which do not lie on the unit circle.

We now consider the case $k \geq 3$.  The tight example is again a
hyperedge of size $k' \leq k$.  The partition function $P_{k'}(z)$ of this
graph is
\begin{displaymath}
  P_{k'}(z) \defeq \beta(1+z)^{k'} + (1-\beta)(1+z^{k'}),
\end{displaymath}
and we will show that it has at least one root outside the unit circle
when $\beta \neq 1$ satisfies
\begin{equation}
  \beta \not\in \inb{-\frac{1}{2^{-k-1} - 1},
    \frac{1}{2^{k-1}\cos^{k-1}\bigl(\frac{\pi}{k-1}\bigr) + 1}}.\label{eq:11}
\end{equation}

We consider three exhaustive cases under which \eqref{eq:11} holds.

\begin{description}
\item[Case 1: $\beta > 1$]  In this case our tight example is a
  hyperedge of size $k' = 2 \leq k$, and the result follows from that
  of the case $k = 2$.

\item[Case 2: $\beta < -\frac{1}{2^{k-1} - 1}$]  In this case, our
  example is a hyperedge of size $k$.  We note then that $P_k(0) = 1$
  and $P_k(1) = 2\beta(2^{k-1} - 1) + 2 < 0$.  Thus, there exists a
  $z$ in the interval $(0, 1)$ for which $P_k(z) = 0$, and hence $P_k$ has
  a zero that is not on the unit circle.

\item[Case 3:
  $\frac{1}{2^{k-1}\cos^{k-1}\inp{\frac{\pi}{k-1}} + 1} < \beta < 1$]
  Our tight example is again a hyperedge of size $k$.  We will show
  that the degree $k$ polynomial $P_{k}$ has at most $k-3$ zeros
  (counting with multiplicities) on the unit circle $C$, and hence must
  have at least one zero outside it.

  We first consider the point $z = -1$.  Note that since
  $\beta \neq 1$, $P_k(-1) = 0$ if and only if $k$ is odd, and in this
  case $P_k'(-1) = k(1-\beta) \neq 0$. Therefore, $-1$ is a zero of
  multiplicity $1$ of $P_k$ when $k$ is odd, and is not a zero of
  $P_k$ when $k$ is even.

  We now consider zeros of $P_k$ in $C \setminus \inbr{-1}$.  Let
  $\tau \defeq 2^{k-1}\frac{\beta}{\beta - 1}$ and
  $g(z) \defeq \frac{1 + z^k}{(1 + z)^k}$.  Note that any
  $z \in C \setminus \inbr{-1}$ is a zero of multiplicity~$l$ of
  $P_k$ if and only if it is a zero of the same multiplicity $l$ of
  the meromorphic function $g(z) - \tau/2^{k-1}$.  In particular, at
  such a $z$, the order of the first non-zero derivative of $P_k$ is
  the same as the order of the first non-zero derivative of $g$, and
  this number is the multiplicity of $z$ as a zero of $P$ (or
  equivalently, as a root of $g(z) = \tau/2^{k-1}$).  Note also that
  $g(z)$ maps $C \setminus \inbr{-1}$ into the real line: in fact, for
  $z = e^{2\iota \theta}$, $\theta \in (-\pi/2, \pi/2)$, we have
  \begin{displaymath}
    2^{k-1}g(z) = 2^{k-1}\cdot\frac{1 + \cos 2k\theta + \iota \sin 2k\theta}{(1 + \cos
      2\theta + \iota \sin 2\theta)^k} = \frac{2^k\cos{k\theta}}{(2
      \cos\theta)^k}\cdot\frac{e^{\iota k \theta}}{e^{\iota k \theta}}
    =
    \frac{\cos k\theta}{\cos^{k}\theta} =:
    h(\theta),
  \end{displaymath}
  and further $h'(\theta) = 2^k\iota z g'(z)$, so that
  $h'(\theta) = 0$ if and only if $g'(z) = 0$.  Indeed, by computing
  further derivatives, one sees that the multiplicity of any root of
  $h(\theta) = \tau$ in $(-\pi/2, \pi/2)$ (i.e., the order of the
  first non-zero derivative of $h$ at the root) is the same as the
  multiplicity of the corresponding root $z = e^{2\iota \theta}$ of
  $g(z) = \tau/2^{k-1}$.

  Thus, in order to establish our claim that $P_k(z)$ has at most
  $k - 3$ zeros (counting with multiplicities and also accounting for
  the possible zero at $-1$ considered above) on the unit circle $C$,
  we only need to show that the number of roots of the equation
  $h(\theta) = \tau$ on $(-\pi/2, \pi/2)$ (counted with
  multiplicities) is at most $k-4$.  We now proceed to establish this
  property of $h$. Note that for the range of $\beta$ being
  considered, we have $\tau < -\sec^{k-1}\bigl(\frac{\pi}{k-1}\bigr)$.

  Since $h(\theta) = h(-\theta)$, we consider its behavior only in the
  interval $I = [0, -\pi/2)$.  We have
  $h'(\theta) = -\frac{k\sin (k-1)\theta}{\cos^{k+1}\theta}$, so that
  the zeros of $h'$ in $I$ are given by $\rho_i = i\pi/(k-1)$, where
  $0 \leq i < \floor{k/2}$ is an integer.  Note that all these zeros
  of $h'$ are in fact simple: $h''(\rho_i) \neq 0$.  Thus, any root of
  $h(\theta) = \tau$ is of multiplicity at most $2$.  Now, define
  $\rho_{\floor{k/2}} = \pi/2$, and let $I_i$ be the interval
  $[\rho_i, \rho_{i+1})$ for $0 \leq i \leq \floor{k/2} - 1$.  We note the
  following facts (see Figure~\ref{fig:1} for an example):
  \begin{enumerate}
  \item In the interval $I_i$, $h$ is strictly decreasing when $i$ is
    even and strictly increasing when $i$ is odd.
  \item For $i < \floor{k/2}$,
    $h(\rho_i) = (-1)^i\sec^{k-1}\bigl(\frac{i\pi}{k-1}\bigr)$, so that $h(\rho_i)$ is
    strictly positive when $i$ is even and strictly negative when $i$
    is odd.  Further, $h(\rho_1) = -\sec^{k-1}\bigl(\frac{\pi}{k-1}\bigr) > \tau$.
  \end{enumerate}
  \begin{figure}[h]
    \centering
    \includegraphics[width=0.7\textwidth]{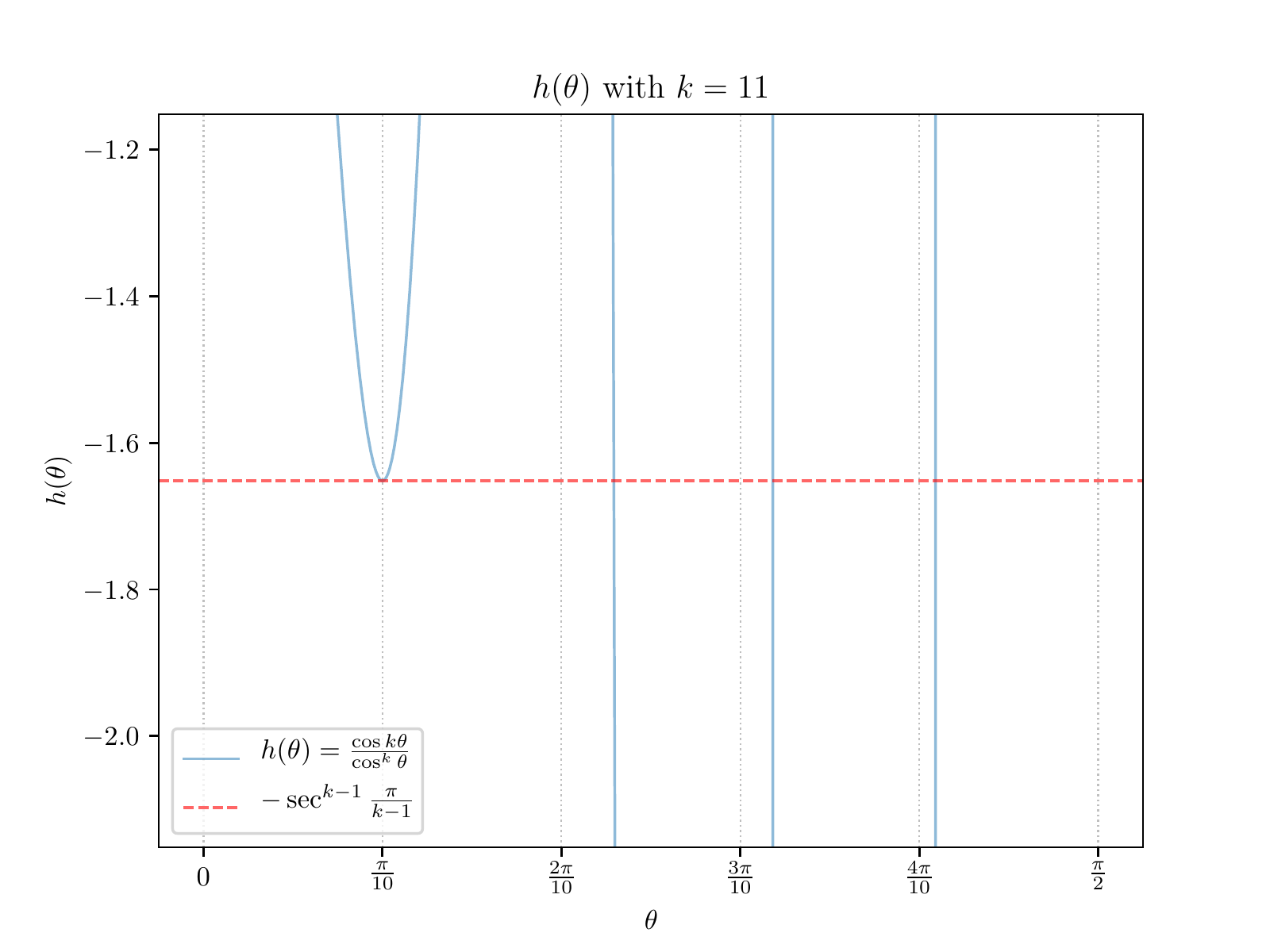}
    \caption{The function $h(\theta)$}
    \label{fig:1}
  \end{figure}
  From these observations we can now deduce that when
  $-\sec^{k-1}\bigl(\frac{\pi}{k-1}\bigr) > \tau$, $h(\theta) = \tau$ has
  \begin{enumerate}
  \item no roots in $I_0$ and $I_1$,
  \item at most two roots in $I_i \cup I_{i+1}$, counting
    multiplicities, when $i$ is a positive even integer strictly less
    than $\floor{k/2} - 1$.  The two roots can arise in only the
    following two ways: there can be one root each, with multiplicity
    1, in each of the two intervals $I_i$ and $I_{i+1}$, or else there
    can be a root of multiplicity $2$ at $\rho_{i+1}$.
  \item at most one additional root in $I_{\floor{k/2} - 1}$, and this
    additional root can arise only when $\floor{k/2} - 1$ is even.
  \end{enumerate}
  Together, the above three items imply that when
  $\tau < -\sec^{k-1}\bigl(\frac{\pi}{k-1}\bigr)$, the number of roots of
  $h(\theta) = \tau$ in $I = [0, -\pi/2)$, counted with their
  multiplicities, is at most $\floor{k/2} - 2$.  Using the symmetry of
  $h$ around $0$ pointed out above, we thus see that the number of
  roots of $h(\theta) -\tau$ in $(-\pi/2, \pi/2)$ is at most $k - 4$,
  so that $P_k$ has at most $k-3$ zeros (accounting for the possible
  simple zero at $-1$ when $k$ is odd) on the unit circle for such
  $\beta$.  This implies that at least one zero of the degree $k$
  polynomial $P_k$ must lie outside the unit circle, as required.
\end{description}
\section*{Acknowledgments}
We thank Alexander Barvinok, Guus Regts and anonymous reviewers for
helpful comments.  JL and AS are supported in part by US NSF grant
CCF-1420934.  PS is supported by a Ramanujan Fellowship of the Indian
Department of Science and Technology. Some of this work was done at
the Simons Institute for the Theory of Computing at UC Berkeley.

\end{document}